\documentclass[11pt]{article}

\usepackage{amsmath}
\usepackage{amsthm}
\usepackage{amssymb}
\usepackage[usenames,dvipsnames]{xcolor}
\usepackage{graphicx}
\usepackage{longtable}
\usepackage[margin=1in]{geometry}
\usepackage[utf8]{inputenc}
\usepackage{scalefnt}
\usepackage[colorlinks=true, allcolors=blue]{hyperref}

\usepackage{algorithm}
\usepackage{algpseudocode}

\usepackage{float}

\usepackage{titlesec}
\titlespacing*{\section}
{0pt}{9pt plus 1pt minus 1pt}{7pt plus 1pt minus 1pt}
\titlespacing*{\subsection}
{0pt}{7pt plus 1pt minus 1pt}{5pt plus 1pt minus 1pt}

\newtheorem{theorem}{Theorem}[section]
\newtheorem{definition}[theorem]{Definition}
\newtheorem{proposition}[theorem]{Proposition}

\newtheorem{lemma}[theorem]{Lemma}

\newtheorem{remark}[theorem]{Remark}

\newtheorem*{conjecture*}{Conjecture}

\newcommand\smaller[2][0.85]{{\scalefont{#1}#2}}

\newcommand{\MAXCUT}{\mbox{\rm \smaller[0.76]{MAX CUT}}}
\newcommand{\MAXAND}{\mbox{\rm \smaller[0.76]{MAX 2-AND}}}
\newcommand{\MAXDICUT}{\mbox{\rm \smaller[0.76]{MAX DI-CUT}}}
\newcommand{\MAXSAT}[1]{\smaller[0.76]{MAX #1-SAT}}
\newcommand{\MAXSATg}{\smaller[0.76]{MAX SAT}}

\newcommand{\MAXNAESATg}{\smaller[0.76]{MAX NAE-SAT}}
\newcommand{\MAXCSP}[1]{\mbox{\rm \smaller[0.76]{MAX #1-CSP}}}
\newcommand{\MAXLIN}[1]{\mbox{\rm \smaller[0.76]{MAX #1-LIN}}}
\newcommand{\MAXCSPP}{\mbox{\rm \smaller[0.76]{MAX CSP}}}

\newcommand{\GE}{\;\ge\;}
\newcommand{\LE}{\;\le\;}
\newcommand{\EQ}{\;=\;}
\newcommand{\GT}{\;>\;}

\title{Separating \MAXAND, \MAXDICUT\ and \MAXCUT}
\author{Joshua Brakensiek\thanks{Stanford University, supported in part by an NSF Graduate Research Fellowship and a Microsoft Research PhD Fellowship. Email: \texttt{jbrakens@cs.stanford.edu}}\and Neng Huang\thanks{University of Chicago, supported in part by NSF grant CCF:2008920. Email: \texttt{nenghuang@uchicago.edu}}\and Aaron Potechin\thanks{University of Chicago, supported in part by NSF grant CCF:2008920. Email: \texttt{potechin@uchicago.edu}}\and Uri Zwick\thanks{Blavatnik School of Computer Science, Tel Aviv University, Israel. Email: \texttt{zwick@tau.ac.il}}}
\date{}

\newcommand{\eps}{\varepsilon}
\newcommand{\E}{\mathop{\mathbb{E}}}

\newcommand{\RR}{\mathbb{R}}

\newcommand{\Inf}{\mathrm{Inf}}

\newcommand{\br}{\mathbf{r}}
\newcommand{\bv}{\mathbf{v}}

\newcommand{\bx}{{\mathbf{x}}}
\newcommand{\by}{{\mathbf{y}}}

\newcommand{\Dist}{\operatorname{Dist}}
\newcommand{\flip}{\operatorname{flip}}

\newcommand{\CUT}{\mbox{\rm \smaller[0.76]{CUT}}}
\newcommand{\AND}{\mbox{\rm \smaller[0.76]{AND}}}
\newcommand{\DICUT}{\mbox{\rm \smaller[0.76]{DI-CUT}}}

\newcommand{\THRESH}{{\cal THRESH}}

\newcommand{\Val}{\mathrm{Val}}

\newcommand{\aGW}{\alpha_{\text{CUT}}}
\newcommand{\agw}{\alpha_{\text{GW}}}
\newcommand{\aDC}{\alpha_{\text{DI-CUT}}}
\newcommand{\aAND}{\alpha_{\text{2AND}}}

\newcommand{\argmin}{\operatorname{argmin}}
\newcommand{\argmax}{\operatorname{argmax}}

\newcommand{\comp}{\mathsf{Completeness}}
\newcommand{\sound}{\mathsf{Soundness}}

\begin{document}

\maketitle

\newcommand{\bestDICUT}{0.87446}
\newcommand{\bestAND}{0.87414}

\newcommand{\bestDICUTverified}{0.87447} %
\newcommand{\bestANDverified}{0.87415}

\vspace*{-20pt}
\begin{abstract}
Assuming the Unique Games Conjecture (UGC), the best approximation ratio that can be obtained in polynomial time for the \MAXCUT\ problem is $\aGW\simeq 0.87856$, obtained by the celebrated SDP-based approximation algorithm of Goemans and Williamson. The currently best approximation algorithm for \MAXDICUT, i.e., the \MAXCUT\ problem in \emph{directed} graphs, achieves a ratio of about $0.87401$, leaving open the question whether \MAXDICUT\ can be approximated as well as \MAXCUT. We obtain a slightly improved algorithm for \MAXDICUT\ and a new UGC-hardness for it, showing that $\bestDICUT\le \aDC\le 0.87461$, where $\aDC$ is the best approximation ratio that can be obtained in polynomial time for \MAXDICUT\ under UGC. The new upper bound separates \MAXDICUT\ from \MAXCUT, i.e., shows that \MAXDICUT\ cannot be approximated as well as \MAXCUT, resolving a question raised by Feige and Goemans.

A natural generalization of \MAXDICUT\ is the \MAXAND\ problem in which each constraint is of the form $z_1\land z_2$, where $z_1$ and $z_2$ are literals, i.e., variables or their negations. (In \MAXDICUT\ each constraint is of the form $\bar{x}_1\land x_2$, where $x_1$ and $x_2$ are variables.) Austrin separated \MAXAND\ from \MAXCUT\ by showing that $\aAND \le 0.87435$ and conjectured that \MAXAND\ and \MAXDICUT\ have the same approximation ratio. Our new lower bound on \MAXDICUT\ refutes this conjecture, completing the separation of the three problems \MAXAND, \MAXDICUT\ and \MAXCUT. We also obtain a new lower bound for \MAXAND\, showing that $\bestAND \le \aAND \le 0.87435$. 

Our upper bound on \MAXDICUT\ is achieved via a simple, analytical proof. The new lower bounds on \MAXDICUT\ and \MAXAND, i.e., the new approximation algorithms, use experimentally-discovered distributions of rounding functions which are then verified via\linebreak computer-assisted proofs.
\footnote{Code for the project: \url{https://github.com/jbrakensiek/max-dicut}}

\end{abstract}
\setcounter{page}{0}
\thispagestyle{empty}

\pagebreak

\section{Introduction}

Goemans and Williamson \cite{GW95}, in their seminal paper, introduced the paradigm of obtaining approximation algorithms for Boolean \emph{Constraint Satisfaction Problems} (CSPs) by first obtaining a semidefinite programming (SDP) \emph{relaxation} of the problem and then \emph{rounding} an optimal solution of the relaxation. The first, and perhaps biggest, success of this paradigm is a simple and elegant $\agw$-approximation algorithm, where $\agw\simeq 0.87856$, for the \MAXCUT\ problem, i.e., the maximum cut problem in undirected graphs, improving for the first time over the naive $\frac{1}{2}$-approximation algorithm. Goemans and Williamson \cite{GW95} also obtained improved algorithms for the \MAXDICUT, \MAXSAT{2} and \MAXSATg\ problems.

Feige and Goemans \cite{FG95}, Matuura and Matsui \cite{MM03} and Lewin, Livnat and Zwick \cite{LLZ02} obtained improved approximation algorithms for the \MAXSAT{2} and \MAXDICUT\ problems. The best approximation ratios, obtained by \cite{LLZ02}, are $0.940$ for \MAXSAT{2} and $0.874$ for \MAXDICUT. Karloff and Zwick \cite{KZ97} obtained an optimal (see below) $\frac{7}{8}$-approximation algorithm for \MAXSAT{\{1,2,3\}} and Zwick~\cite{Zwick98} obtained approximation algorithms, some of them optimal, for many other \MAXCSP{3} problems, i.e., maximization versions of Boolean CSP problems in which each constraint is on at most three variables.
Andersson and Engebretsen~\cite{AE98}, Zwick~\cite{Zwick99a}, Halperin and Zwick~\cite{HZ01}, Asano and Williamson~\cite{AW02}, Zhang, Ye and Han~\cite{ZYH04}, and Avidor, Berkovitch and Zwick~\cite{ABZ05} obtained approximation algorithms for various versions of the \MAXSATg\ and \MAXNAESATg\ problems. It is a major open problem whether there is a $\frac78$-approximation algorithm for the MAX SAT problem. \cite{BHPZ21} showed that there is no $\frac78$-approximation algorithm for the \MAXNAESATg\ problem, assuming UGC. \cite{ABGNS22} and \cite{EN19} used ``sticky Brownian motion'' to obtain optimal, or close to optimal, algorithms for \MAXCUT\ and related problems. For a survey of these and related results, see Makarychev and Makarychev \cite{MM17}.

H{\aa}stad \cite{H01}, in a major breakthrough, extending the celebrated PCP theorem of~\cite{ALMSS98}, showed, among other things, that, for any $\eps>0$, it is NP-hard to obtain a $(\frac{7}{8}+\eps)$-approximation of \MAXSAT{3} and a $(\frac{1}{2}+\eps)$-approximation of \MAXLIN{3}, showing that the trivial algorithms for these two problems that just choose a random assignment are tight. \cite{TSSW00} showed, using gadget reductions, that it is NP-hard to obtain a $(\frac{16}{17}+\eps)$-approximation of \MAXCUT\ and $(\frac{12}{13}+\eps)$-approximation of \MAXDICUT.

Khot \cite{khot02} introduced the \emph{Unique Games Conjecture} (UGC). Khot, Kindler, Mossel and O'Donnell \cite{KKMO07} then showed that UGC implies that, for any $\eps>0$, obtaining an $(\agw+\eps)$-approximation for \MAXCUT\ is NP-hard, showing, quite remarkably, that the algorithm of Goemans and Williamson \cite{GW95} is optimal, i.e., $\aGW=\agw$, assuming UGC. Austrin \cite{Austrin07} then showed that the \MAXSAT{2} algorithm of Lewin, Livnat and Zwick \cite{LLZ02} is essentially optimal, again modulo UGC. Austrin \cite{Austrin10} obtained some upper bounds on the the approximation ratio that can be achieved for \MAXAND\ in polynomial time. However, they do not match the approximation ratio obtained by the \MAXDICUT\ algorithm of~\cite{LLZ02} which is in fact an approximation algorithm for \MAXAND. %

Raghavendra \cite{R08,R09}, in another breakthrough, showed that under UGC, the best approximation ratio that can be obtained for any MAX CSP problem, over a finite domain and with a finite number of constraint types, can be obtained using a canonical SDP relaxation of the problem and the rounding of an optimal solution of this relaxation using an appropriate rounding procedure taken from a specified family of rounding procedures. The approximation ratio obtained is then exactly the \emph{integrality gap} of the relaxation. Approximating the integrality gap up to $\eps$ takes doubly exponential time in $1/\eps$, and a close to optimal algorithm can be obtained by trying discretized versions of all rounding procedures, up to some resolution. (See Raghavendra and Steurer \cite{RS09} for more on finding almost optimal rounding schemes.) 

It might seem that these results resolve all problems related to the approximation of MAX CSP problems. Unfortunately, this is not the case. These results do give valuable guidance to the designers of approximation algorithms. In particular, it is clear which semidefinite programming relaxation should be used and the search for an optimal, or almost optimal, rounding procedure can be restricted to the family of rounding procedures specified by Raghavendra \cite{R08}. However, these results give almost no concrete information on the integrality gap of the relaxation, which is also the best approximation ratio that can be obtained. Also, no practical information is given on how to obtain optimal, or almost optimal rounding procedures, other than the fact that they belong to a huge class of rounding procedures, as it is wildly impractical to implement and run a brute force algorithm whose running time is doubly exponential in $1/\eps$.

In particular, Raghavendra's results are unable\footnote{In particular, if the answer to any of these questions is ``yes'' the Raghavendra-Steurer algorithm cannot certify these in finite time, and if the answer is ``no'' the $\eps$ needed for separation is so small that the algorithm would need to run over a galactic time scale.} to answer questions of the following form: Is there a $\frac{7}{8}$-approximation algorithm for \MAXSATg, with clauses of all sizes allowed? Can \MAXDICUT\ be approximated as well as \MAXCUT? Can \MAXAND\ be approximated as well as \MAXDICUT? In this paper we study the latter two questions and answer them in the negative, assuming UGC.

\subsection{Our results} 

Our main result is the following theorem.
\begin{theorem}[Main]
Assuming UGC, $\aAND < \aDC < \aGW$.
\end{theorem}

To separate \MAXAND, \MAXDICUT\ and \MAXCUT, we obtain an improved upper bound and an improved lower bound (i.e., an approximation algorithms) for \MAXDICUT. Our improved upper bound is: %

\begin{theorem}\label{theorem:upper}
Assuming UGC, $\aDC \leq 0.87461$.
\end{theorem}

To obtain the new upper bound, we construct a distribution over \MAXDICUT\ \emph{configurations} that is hard for any rounding procedure from the family $\THRESH^{-}$ defined by Lewin, Livnat and Zwick~\cite{LLZ02}. Such hard distributions can then be converted into \emph{dictatorship tests} and then Unique Games hardness by small modifications to the technique used by Austrin~\cite{Austrin10} for distributions over \MAXAND\ configurations. 

It is more difficult to obtain hard configurations for \MAXDICUT\ than for \MAXAND, since in \MAXAND\ the functions used in the rounding scheme can be assumed, without loss of generality, to be \emph{odd}. (A function $f:[-1,1]\to \RR$ is odd if and only if $f(-x)=-f(x)$ for every $x\in[-1,1]$.) Using an odd rounding function ensures that a variable and its negation are assigned opposite truth values. In \MAXDICUT\ there is no such restriction as, in a sense, there are no negated variables. The possibility of using non-odd rounding functions gives the rounding scheme more power. (The improved rounding scheme that we obtain for \MAXDICUT\ uses a distribution of rounding functions some of which are not odd. This is exactly what enables the separation of \MAXDICUT\ from \MAXAND, as we discuss below.) We overcome this difficulty using a simple, symmetric construction for which the best rounding scheme is odd. Another interesting feature of our hard construction is that it contains a configuration for which all the triangle inequalities, powerful constraints of the SDP relaxation, are not tight. This is in contrast to previous work on \MAXSAT{2}~\cite{Austrin07} and \MAXAND~\cite{Austrin10}, where hardness results are derived only from configurations in which one of the triangle inequalities is tight. 

Our construction yields an upper bound of $\aDC\le 0.87461$, 
which together with $\aGW\ge 0.87856$ exhibits a clear separation between \MAXDICUT\ and \MAXCUT. (Although the separation is clear, it is still perplexing that the approximation ratios of \MAXCUT\ and \MAXDICUT\ are so close, and yet not equal.) We believe that our upper bound can be slightly improved using a sequence of more and more complicated constructions that yield slightly better and better upper bounds. %

In addition to our improved upper bound for \MAXDICUT, we also obtain two new lower bounds for \MAXDICUT\ and \MAXAND.

\begin{theorem}\label{theorem:lower} $\aDC\ge \bestDICUT$. (In other words, there is an approximation algorithm for \MAXDICUT\ with an approximation ratio of at least $\bestDICUT$.)
\end{theorem}

\begin{theorem}\label{theorem:2and} $\aAND\ge \bestAND$. (In other words, there is an approximation algorithm for \MAXAND\ with an approximation ratio of at least $\bestAND$.)
\end{theorem}

The new lower bounds 
improve on the previously best, and non-rigorous, bound of $0.87401$ obtained by~\cite{LLZ02} for both \MAXDICUT\ and \MAXAND. Despite the relatively small improvements, the improved approximation algorithms are interesting for at least two reasons. The first is that the new approximation algorithm for \MAXDICUT\ separates \MAXDICUT\ from \MAXAND, refuting a conjecture of Austrin~\cite{Austrin10}. The second is that the new algorithms show that %
taking a single rounding scheme from $\THRESH^{-}$,
as done by~\cite{LLZ02} and as shown by Austrin~\cite{Austrin07, Austrin10} to be sufficient for obtaining an optimal approximation algorithm for \MAXSAT{2}, is not sufficient for obtaining optimal approximation algorithms for \MAXDICUT\ and \MAXAND.   
Using insights gained from the upper bounds, we design an improved approximation algorithm for \MAXDICUT\ that uses \emph{distributions} of $\THRESH^{-}$ rounding procedures, i.e., rounding procedures belonging to the more general family $\THRESH$ also defined in \cite{LLZ02}. Using a computer search, we find a new rounding procedure from this family\footnote{Technically, we add in a tiny amount of independent rounding for verification purposes but we believe this can be removed.} which shows that $\aDC \geq \bestDICUT$. Our rigorous proof of this inequality is computer assisted. %

In \cite{LLZ02}, the authors discovered their $\THRESH^{-}$ procedures for \MAXDICUT\ and \MAXAND\ using non-convex optimization. More precisely, they used a local descent procedure from random starting points to tune a single rounding function that performs well for all possible configurations simultaneously. However, this approach becomes impractical for finding an optimal probability distribution of $\THRESH^{-}$ functions (i.e., a ``$\THRESH$ scheme''). One potential reason why this would not work is that there would be a significant local optimum where all the functions in the distribution identical to the one in \cite{LLZ02}.

Instead, we cast the design of the $\THRESH$ scheme for \MAXDICUT\ and \MAXAND\ as infinite \emph{zero-sum games} played by two players. The first player, Alice, selects a $\THRESH^{-}$ function and the second player, Bob, selects a configuration of SDP vectors to round. (This configuration may or may not correspond to an optimal solution of an SDP relaxation of an actual instance.)
Alice's value is then the approximation ratio achieved by her $\THRESH^{-}$ function on the SDP value of the configuration. Bob's value is the negative of Alice's value. One can then show, that $\aDC$ (or $\aAND$) is precisely the value of this game, assuming UGC and the positivity conjecture in~\cite{Austrin10}. Computationally, we discretize this game and use a min-max optimization procedure to estimate the value of this game and find an optimal, or almost optimal, strategy for Alice. This proceeds in a series of phases: Bob challenges Alice with a distribution of instances, and Alice computes a nearly-optimal response using methods similar to that of \cite{LLZ02}. Then, with Alice's functions, Bob computes a new distribution of instances which Alice does the worst one. This latter step is done by solving a suitable LP (the dual variables tell us Alice's optimal $\THRESH$ scheme). As the ``raw'' $\THRESH$ scheme produced by this procedure can be somewhat noisy, we subsequently manually simplified the $\THRESH$ distribution.

As mentioned, the proofs of the bounds $\aDC\geq \bestDICUT$ and $\aAND\ge \bestAND$ are computer-assisted, using the technique of \emph{interval arithmetic}. This technique has been previously used in the study of approximation algorithms. For example, Zwick~\cite{zwick02} used it to certify the $\frac{7}{8}$-approximation ratio for \MAXSAT{\{1,2,3\}} claimed by~\cite{KZ97}. The use of interval arithmetic in our setting is much more challenging, however, as the rounding procedures used for \MAXDICUT\ are much more complicated than the simple random hyperplane rounding used for \MAXSAT{\{1,2,3\}}. In particular, we need to use rigorous numerical integration to compute two-dimensional normal probabilities. A computer-assisted verification is probably necessary in our setting since fairly complicated distributions seem to be need for obtaining good approximation ratios, and it is hard to imagine that such distributions can be analyzed manually.

Since Austrin~\cite{Austrin10} showed that $\aAND < 0.87435$, assuming UGC, our new \MAXDICUT\ approximation algorithm separates \MAXAND\ and \MAXDICUT. This refutes Austrin's conjecture that \MAXAND\ and \MAXDICUT\ have the same approximation ratios. It also gives an interesting, non-trivial, example where a positive CSP (i.e., CSP that does not allow negated variables) is strictly easier to approximate than the CSP with the same predicate when negated variables are allowed.

We believe that the fact that rounding procedures from $\THRESH^-$ do not yield optimal approximation algorithms for \MAXDICUT\ is interesting in its own right. We conjecture that distributions over such procedures, i.e., rounding procedures from $\THRESH$ are enough to obtain optimal algorithms for \MAXDICUT\ and \MAXAND. (A continuous distribution is probably needed to get the optimal algorithms.) %

We note that both $\THRESH^-$ and $\THRESH$ are tiny subfamilies of the families shown by Raghavendra \cite{R08} to be enough for obtaining optimal approximation algorithms for general \MAXCSPP\ problems. In particular, $\THRESH^-$ and $\THRESH$ use only one Gaussian random vector while, in general, the families of Raghavendra \cite{R08} may need an unbounded number of such random vectors to obtain optimal or close-to-optimal results. %

\subsection{Organization of paper}

The rest of the paper is organized as follows. In Section~\ref{sec:prelim} we introduce the \MAXCUT, \MAXDICUT\ and \MAXAND\ problems and their SDP relaxations, we state the Unique Games Conjecture, and introduce the $\THRESH^-$ and $\THRESH$ families of rounding procedures used throughout the paper. In Section~\ref{sec:upper} we derive our new upper bound on \MAXDICUT\ which separates \MAXDICUT\ from \MAXCUT. The proof of this upper bound is completely analytical. In Section~\ref{sec:lower} we describe the computation techniques used to discover our improved \MAXDICUT\ algorithm and the computation techniques used to rigorously verify the approximation ratio that it achieves. In Section~\ref{sec:lower-AND} we obtain corresponding results for the \MAXAND\ problem. We end in Section~\ref{sec:concl} with some concluding remarks and open problems.

\section{Preliminaries}\label{sec:prelim}

\subsection{MAX CSP and canonical SDP relaxations}
For a Boolean variable, we associate $-1$ with true and 1 with false. A Boolean predicate on $k$ variables is a function $P: \{-1, 1\}^k \to \{0, 1\}$. If $P$ outputs 1, then we say $P$ is satisfied. 

\begin{definition}[$\MAXCSPP(P)$]
Let $P$ be a Boolean predicate on $k$ variables. An instance of $\MAXCSPP(P)$ is defined by a set of Boolean variables $\mathcal{V} = \{x_1, x_2, \ldots, x_n\}$ and a set of constraints $\mathcal{C} = \{C_1, C_2, \ldots, C_m\}$, where each constraint $C_i$ is of the form 
$P(b_{i,1}x_{j_{i,1}}, b_{i,2}x_{j_{i,2}}, \ldots, b_{i,k}x_{j_{i,k}})$ for some $j_{i,1}, \ldots, j_{i,k} \in [n]$ and $b_{i,1}, b_{i,2}, \ldots b_{i,k} \in \{-1, 1\}$,
and a weight function $w: \mathcal{C} \to [0, 1]$ satisfying $\sum_{i = 1}^m w(C_i) = 1$. The goal is to find an assignment to the variables that maximizes $\sum_{i = 1}^m w(i)P(b_{i,1}x_{j_{i,1}}, b_{i,2}x_{j_{i,2}}, \ldots, b_{i,k}x_{j_{i,k}})$, i.e., the sum of the weights of satisfied constraints.
\end{definition}

\begin{definition}[$\MAXCSPP^+(P)$]
$\MAXCSPP^+(P)$ has the same definition as $\MAXCSPP(P)$, except that now each constraint $C_i$ is of the form $P(x_{j_{i,1}}, x_{j_{i,2}}, \ldots, x_{j_{i,k}})$. In other words, negated variables are not allowed.
\end{definition}

Since the weight function is non-negative and sums up to 1, we can think of it as a probability distribution over the constraints. Note that we only defined CSPs with a single Boolean predicate, while in general there can be more than one predicate and they may not be Boolean. We refer to a CSP with a $k$-ary predicate as a $k$-CSP.

We are now ready to define the three MAX 2-CSP problems that we separate.

\begin{definition}
Let $\CUT: \{-1, 1\}^2 \to \{0, 1\}$ be the predicate which is satisfied if and only if the two inputs are not equal. Let $\DICUT: \{-1, 1\}^2 \to \{0, 1\}$ be the predicate which is satisfied if and only if $x = 1$ and $y = -1$. Then \MAXCUT\ is the problem $\MAXCSPP^+(\CUT)$, \MAXDICUT\ is the problem $\MAXCSPP^+(\DICUT)$ and \MAXAND\ is the problem $\MAXCSPP(\DICUT)$. 
\end{definition}

In graph-theoretic language, we can think of each variable in a \MAXDICUT\ instance as a vertex, and each constraint as a weighted direct edge between two vertices. An assignment of $+1$ and $-1$ to the vertices defines a directed cut in the graph. We are asked to assign $+1$ and $-1$ to the vertices so that the sum of the weights of edges that cross the cut, i.e., go from $+1$ to $-1$, is maximized.

We can also define $\AND: \{-1, 1\}^2 \to \{0, 1\}$ such that $\AND(x,y)=1$ if and only if $x=y=-1$. Note that then $\DICUT(x,y)=\AND(\bar{x},y)$, and \MAXAND\ is also $\MAXCSPP(\AND)$, hence its name.

The following Fourier expansion of $\DICUT$ is heavily used throughout the paper.
\begin{proposition}
$\DICUT(x, y) = \frac{1 + x - y - xy}{4}$.
\end{proposition}
This proposition can be used to extend the domain of $\DICUT$ to real inputs.

Any $\MAXCSPP(P)$ has a \emph{canonical} semi-definite programming relaxation.
The canonical SDP relaxation for \MAXDICUT, for example, is:
\begin{alignat*}{3}
&\text{maximize} &\qquad& \sum_{C={\smaller[0.76] \DICUT}(x_{i}, x_{j}) \in \mathcal C}  w_C \cdot \frac{1 + \bv_0\cdot\bv_{i} - \bv_0\cdot\bv_{j} - \bv_{i}\cdot\bv_{j}}{4}\\
&\text{subject to} &\qquad& \forall i \in \{0, 1, 2, \ldots, n\},\ \ \ \qquad\,\, \bv_i \cdot \bv_i = 1,\\
& &\qquad& \forall C = \DICUT(x_{i}, x_{j}) \in \mathcal{C}, 
\begin{array}{c}
(\bv_0 - \bv_i) \cdot (\bv_0 - \bv_j) \GE 0,\\
(\bv_0 + \bv_i) \cdot (\bv_0 - \bv_j) \GE 0,\\
(\bv_0 - \bv_i) \cdot (\bv_0 + \bv_j) \GE 0,\\
(\bv_0 + \bv_i) \cdot (\bv_0 + \bv_j) \GE 0.\\
\end{array}
\end{alignat*}
The canonical SDP relaxation is obtained as follows. There is a unit vector $\bv_i \in \mathbb{R}^{n+1}$ for each variable $x_i$, and a special unit vector $\bv_0$ corresponding to false. Each linear term $x_i$ in the Fourier expansion of $P$ is replaced by $\bv_0 \cdot \bv_i$, and each quadratic term $x_ix_j$ is replaced by $\bv_i\cdot \bv_j$. The so-called triangle inequalities are then added.

Note that this is the special case of Raghavendra's basic SDP in the setting of Boolean 2-CSPs, and the triangle inequalties ensure that there is a local distribution of assignments for each constraint. 

\subsection{Unique Games Conjecture}
The Unique Games Conjecture (UGC), introduced by Khot~\cite{khot02}, plays a crucial role in the study of hardness of approximation of CSPs. One version of the conjecture is as follows.

\begin{definition}[Unique Games]
In a unique games instance $I = (G, L, \Pi)$, we are given a weighted graph $G = (V(G), E(G), w)$, a set of labels $[L] = \{1, 2, \ldots, L\}$ and a set of permutations $\Pi = \{\pi_e^v : [L] \to [L] \mid e = \{v, u\} \in E(G)\}$ such that for every $e = \{u, v\} \in E(G)$, $\pi_e^{v} = (\pi_e^{u})^{-1}$. An assignment to this instance is a function $A: V(G) \to [L]$. We say that $A$ satisfies an edge $e = \{u, v\}$ if $\pi_e^u(A(u)) = A(v)$. The value of an assignment $A$ is the weight of satisfied edges, i.e., $\Val(I, A) = \sum_{e \in E(G): A \textrm{ satisfies } e} w(e)$, and the value of the instance $\Val(I)$ is defined to be the value of the best assignment, i.e., $\Val(I) = \max_A \Val(I, A)$.
\end{definition}

\begin{conjecture*}[Unique Games Conjecture]
For any $\eta, \gamma > 0$, there exists a sufficiently large $L$ such that the problem of determining whether a given unique games instance $I$ with $L$ labels has $\Val(I) \geq 1 - \eta$ or $\Val(I) \leq \gamma$ is NP-hard.
\end{conjecture*}

We say that a problem is UG-hard, if it is NP-hard assuming the UGC. Raghavendra \cite{R08} showed that any integrality gap instance of the canonical SDP relaxation can be turned into a UG-hardness result.

\subsection{Configurations of biases and pairwise biases}
As it turns out, an actual integrality gap instance is not required to derive UG-hardness results. Instead, it is sufficient to consider configurations of SDP solution vectors that appear in the same constraint. For 2-CSPs, each such configuration is represented by a triplet $\theta = (b_i, b_j, b_{ij})$, where $b_i = \bv_0 \cdot \bv_i$, and $b_j = \bv_0 \cdot \bv_j, b_{ij} = \bv_i \cdot \bv_j$. $b_i$ and $b_j$ are called \emph{biases} and $b_{ij}$ is called a \emph{pairwise bias}. A valid configuration is required to satisfy the triangle inequalities described in the previous section. As long as the triangle inequalities are satisfied, it does not matter whether such a configuration comes from an actual SDP solution. We will use $\Theta$ for a set of valid configurations, and $\Tilde{\Theta}$ for such a set endowed with a probability distribution.

\begin{definition}[Completeness]
Given a configuration $\theta = (b_i, b_j, b_{ij})$ for \MAXDICUT, its completeness is defined as $\comp(\theta)=\frac{1 + b_i - b_j - b_{ij}}{4}$. For a distribution of configurations $\Tilde{\Theta}$, its completeness  is defined as $\comp(\Tilde{\Theta})=\mathbb{E}_{\theta \sim \Tilde{\Theta}}[\comp(\theta)]$.
\end{definition}
Note that if $\Tilde{\Theta}$ actually comes from an SDP solution, then $\comp(\Tilde{\Theta})$ is simply the SDP value of this solution.

\begin{definition}[Relative pairwise bias]
Given a configuration $\theta = (b_i, b_j, b_{ij})$, the relative pairwise bias is defined as $\rho(\theta)=\frac{b_{ij} - b_ib_j}{\sqrt{(1 - b_i^2)(1 - b_j^2)}}$, if $(1 - b_i^2)(1 - b_j^2)\neq 0$, and $0$ otherwise.
\end{definition}
Geometrically, $\rho(\theta)$ is the inner product between $\bv_i$ and $\bv_j$ after removing their components parallel to $\bv_0$ and renormalizing. 

\begin{definition}[Positive configurations~\cite{Austrin10}]
Given a Boolean predicate $P(x_1, x_2)$ on two variables with Fourier expansion $\frac{\hat{P}_\emptyset + \hat{P}_1x_1+ \hat{P}_2x_2+ \hat{P}_{1, 2}x_1x_2}{4}$, a configuration $\theta = (b_i, b_j, b_{ij})$ for $\MAXCSPP(P)$ (or $\MAXCSPP^+(P)$) is called \emph{positive} if $\hat{P}_{1, 2}\cdot \rho(\theta) \geq 0$.
\end{definition}

If $P = \DICUT$, then the quadratic coefficient in the Fourier expansion is $-1/4$, which implies that a configuration is positive if and only if its relative pairwise bias is not positive. Austrin~\cite{Austrin10} presented a general mechanism to deduce UG-hardness results for $\MAXCSPP(P)$ from hard distributions of positive configurations. With very slight modifications, the same mechanism can also be used for $\MAXCSPP^+(P)$. Austrin also conjectured that positive configurations are the hardest to round. This conjecture is still open. Our results do not rely on this conjecture. 

In a \MAXDICUT\ instance, if we flip the direction of every edge in the graph, then an optimal solution to this new instance can be obtained by flipping all the signs in an optimal solution to the original instance. For configurations, this symmetry corresponds to swapping the two biases and then changing the signs.

\begin{definition}[Flipping a configuration]
Let $\theta = (b_i, b_j, b_{ij})$ be a $\DICUT$ configuration. We define its \emph{flip} to be $\flip(\theta)=(-b_j, -b_i, b_{ij})$.
\end{definition}

The following proposition can be easily verified.
\begin{proposition}\label{prop:flip_comp}
Let $\theta = (b_i, b_j, b_{ij})$ be a $\DICUT$ configuration. We have
\begin{enumerate}
\item $\rho(\theta) = \rho(\flip(\theta))$.
\item $\comp(\theta) = \comp(\flip(\theta))$.
\end{enumerate}
\end{proposition}

\subsection{The \texorpdfstring{$\THRESH$}{THRESH} and \texorpdfstring{$\THRESH^-$}{THRESH-} families of rounding functions}

$\THRESH$ and $\THRESH^-$, first introduced in~\cite{LLZ02}, are small but powerful families of rounding functions for SDP relaxations of CSPs. In a $\THRESH^-$ rounding scheme, a continuous threshold function $f: [-1, 1] \to \mathbb{R}$ is specified. The algorithm chooses a random Gaussian vector $\br \in \mathbb{R}^{n + 1}$, and sets each variable $x_i$ to true ($-1$) if and only if $\br \cdot \bv_i^\perp \geq f(\bv_0 \cdot \bv_i)$, where 
\[
\bv_i^\perp \;=\; \frac{\bv_i - (\bv_i \cdot \bv_0)\bv_0}{\sqrt{1 - (\bv_i \cdot \bv_0)^2}}
\]
is the component of $\bv_i$ orthogonal to $\bv_0$ renormalized to a unit vector. (If $\bv_i = \pm \bv_0$, we can take~$\bv_i^\perp$ to be any unit vector that is orthogonal to every other vector in the SDP solution.) Since $\bv_i^\perp$ is a unit vector, $\br \cdot \bv_i^\perp$ is a standard normal random variable. Furthermore, for any $i, j\in[n]$, $\br \cdot \bv_i^\perp$ and $\br \cdot \bv_j^\perp$ are jointly Gaussian with correlation $\bv_i^\perp \cdot \bv_j^\perp$.

Let $\Phi, \varphi$ be the c.d.f. and p.d.f. of the standard normal distribution, respectively. For $t_1, t_2 \in \mathbb{R}$, let $\Phi_\rho(t_1, t_2) := \Pr[X \leq t_1 \wedge Y \leq t_2]$, where $X$ and $Y$ are two standard normal random variables that are jointly Gaussian with $\mathbb{E}[XY] = \rho$. Then for a $\THRESH^-$ rounding scheme with threshold function $f$, a variable $x_i$ is rounded to false with probability $\Phi(f(b_i))$. For a $\DICUT$ configuration $\theta=(b_i, b_j, b_{ij})$, the probability that it is satisfied by $\THRESH^-$ with $f$, which happens when $x_i$ is set to false and $x_j$ is set to true, is equal to 
\begin{align*}
\Pr\left[\br \cdot \bv_i^\perp \leq f(b_i) \text{ and } \br \cdot \bv_j^\perp \geq f(b_j)\right] &  \;=\; \Pr\left[\br \cdot \bv_i^\perp \leq f(b_i) \text{ and } -\br \cdot \bv_j^\perp \leq -f(b_j)\right] \\
& \;=\; \Phi_{-\rho(\theta)}(f(b_i), -f(b_j))\;.
\end{align*}
 This naturally leads to the following definition.

\begin{definition}[Soundness]
Let $f:[-1, 1] \to \mathbb{R}$ be a continuous threshold function and $\theta = (b_i, b_j, b_{ij})$ a configuration for \MAXDICUT. We define $\sound(\theta, f) = \Phi_{-\rho(\theta)}(f(b_i), -f(b_j))$. For a distribution of configurations $\Tilde{\Theta}$, its soundness $\sound(\Tilde{\Theta}, f)$ is defined as $\mathbb{E}_{\theta \sim \Tilde{\Theta}}[\sound(\theta, f)]$. 
\end{definition}

As in the case for configurations, we can also flip a $\THRESH^-$ threshold function.

\begin{definition}
Let $f: [-1, 1] \to \mathbb{R}$ be a continuous threshold function. We define $\flip(f)$ as the function $x \mapsto -f(-x)$.
\end{definition}

\begin{proposition}\label{prop:flip_sound}
Let $f: [-1, 1] \to \mathbb{R}$ be a continuous threshold function and $\theta = (b_i, b_j, b_{ij})$ a configuration. Then
\[
\sound(\theta, f) \EQ \sound(\flip(\theta), \flip(f))\;.
\]
\end{proposition}
\begin{proof}
By Proposition~\ref{prop:flip_comp}, we have that $\rho(\theta) = \rho(\flip(\theta)) = \rho$. By definition of soundness, 
\begin{align*}
\sound(\theta, f) & \EQ \Phi_{-\rho}(f(b_i), -f(b_j)) \\
& \EQ \Phi_{-\rho}(-\flip(f)(-b_i), \flip(f)(-b_j)) \\
& \EQ \Phi_{-\rho}(\flip(f)(-b_j), -\flip(f)(-b_i)) \\
& \EQ \sound(\flip(\theta), \flip(f))\;. \qedhere
\end{align*}
\end{proof}

A rounding scheme from $\THRESH$ can be thought of as a distribution over $\THRESH^-$ rounding schemes. Formally speaking, a $\THRESH$ rounding scheme is specified by a continuous function $T: \mathbb{R} \times [-1, 1] \to \mathbb{R}$, and a variable $x_i$ is set to true if and only if $\br \cdot \bv_i^\perp \geq T(\bv_0 \cdot \br, \bv_0 \cdot \bv_i)$. This allows for a continuous distribution over $\THRESH^-$ rounding schemes.

The following partial derivatives are helpful for analyzing $\THRESH$ and $\THRESH^-$ rounding schemes.

\begin{proposition}[Partial derivatives of $\Phi_\rho(t_1,t_2)$]\label{prop:Phi_partial}
\begin{align*}
\frac{\partial \Phi_\rho(t_1, t_2)}{\partial \rho} & \EQ \frac{1}{2\pi \sqrt{1 - \rho^2}} \exp\left(-\frac{t_1^2 - 2\rho t_1t_2 + t_2^2}{2(1 - \rho^2)}\right)\;, \\
\frac{\partial \Phi_\rho(t_1, t_2)}{\partial t_1} & \EQ \varphi(t_1)\Phi\left(\frac{t_2 - \rho t_1}{\sqrt{1 - \rho^2}}\right)\;, \\
\frac{\partial \Phi_\rho(t_1, t_2)}{\partial t_2} & \EQ \varphi(t_2)\Phi\left(\frac{t_1 - \rho t_2}{\sqrt{1 - \rho^2}}\right)\;. 
\end{align*}
\end{proposition}
A derivation of the formula given for $\frac{\partial \Phi_\rho(t_1, t_2)}{\partial \rho}$ can be found in Drezner and Wesolowsky \cite{DW90}. The formulas for $\frac{\partial \Phi_\rho(t_1, t_2)}{\partial t_1}$ and $\frac{\partial \Phi_\rho(t_1, t_2)}{\partial t_2}$ follow easily from the definition of $\Phi_\rho(t_1,t_2)$.

\section{Upper bounds for \texorpdfstring{\MAXDICUT}{MAX DI-CUT}}\label{sec:upper}

\subsection{Separating \texorpdfstring{\MAXDICUT}{MAX DI-CUT} from \texorpdfstring{\MAXCUT}{MAX CUT}}

In this section, we prove the following theorem, which separates \MAXDICUT\ from \MAXCUT.

\begin{theorem}\label{thm:dicut_upper}
Assuming the Unique Games Conjecture, it is NP-hard to approximate \MAXDICUT\ within a factor of $0.87461$.
\end{theorem}

To prove Theorem~\ref{thm:dicut_upper} we construct a distribution of positive configurations $\Tilde{\Theta}$, compute its completeness, and show that no $\THRESH^-$ rounding scheme can achieve a performance ratio of 0.87461 on it. The UG-hardness result then follows from a slight generalization of a reduction of Austrin~\cite{Austrin10}. \footnote{Since the distribution is fixed, the optimal $\THRESH^-$ rounding scheme for it is also the best $\THRESH$ rounding scheme.} (For completeness, we describe this reduction in Appendix~\ref{A-reduction}.)

The distribution $\Tilde{\Theta}$ used to obtain the upper bound is extremely simple. 
Let $p_1, p_2, b, c$ be some parameters to be chosen later. We will choose them so that $b, p_1, p_2 \in (0, 1)$, $c \in (-1, -b^2)$, and $2p_1 + p_2 = 1$. 
Consider the following distribution of configurations $\Tilde{\Theta} = \{\theta_1, \theta_2, \theta_3\}$:
\begin{center}
\begin{tabular}{lc}
    
    $\theta_1 \EQ (-b, -b, -1 + 2b)$ & \text{with probability} $p_1$ \\
    
    $\theta_2 \EQ (\phantom{-}b, -b, \hspace*{17pt}c\hspace*{17pt})$ & \text{with probability} $p_2$ \\
   
    $\theta_3 \EQ (\phantom{-}b, \phantom{-}b, -1 + 2b)$ & \text{with probability} $p_1$ \\
    
\end{tabular}
\end{center}

Note that in the $\theta_1$ and $\theta_3$ one of the triangle inequalities is tight, while in $\theta_2$ none of the triangle inequalities are tight, as was mentioned earlier. Also, this distribution is symmetric with respect to $\flip$, since $\flip(\theta_1) = \theta_3$ and $\flip(\theta_2) = \theta_2$.

We first verify that $\Tilde{\Theta}$ satisfies the positivity condition. 

\begin{proposition}\label{prop:positive}
$\Tilde{\Theta}$ is a distribution of positive configurations.
\end{proposition}
\begin{proof}
In $\theta_1$ and $\theta_3$, the relative pairwise bias is equal to 
$
\rho_1 = \frac{-1 + 2b - b^2}{1 - b^2} = -\frac{1-b}{1+b} < 0.
$
In $\theta_2$, the relative pairwise bias is equal to 
$
\rho_2 = \frac{c + b^2}{1 - b^2} < 0
$
since we choose $c < -b^2$.
\end{proof}

The completeness of this instance can be easily computed.
\begin{proposition}
$\displaystyle\quad
\comp(\Tilde{\Theta}) = p_1\cdot(1-b) + p_2\cdot \frac{1 + 2b - c}{4}.
$
\end{proposition}
\begin{proof}
We have
\begin{align*}
    &\,\, \comp(\Tilde{\Theta}) \\
     = &\,\, p_1 \cdot \frac{1 + (-b) - (-b) - (-1 + 2b)}{4} + p_2 \cdot \frac{1+b-(-b)-c}{4} +  p_1 \cdot \frac{1 + b - b - (-1 + 2b)}{4} \\
     = &\,\, p_1 \cdot \frac{2 - 2b}{4} + p_2 \cdot \frac{1+2b-c}{4} +  p_1 \cdot \frac{2 - 2b}{4} \\
     = &\,\, p_1\cdot(1-b) + p_2\cdot \frac{1 + 2b - c}{4}. \qedhere
\end{align*}
\end{proof}

We now give an upper bound on the performance of any $\THRESH^-$ rounding scheme on this distribution. Let $t_1, t_2$ be the thresholds for $-b, b$ respectively. Let $s(t_1, t_2)$ be the soundness of this rounding scheme on $\Tilde{\Theta}$. By definition of $\THRESH^-$, we have 
\[
s(t_1, t_2) \EQ p_1 \cdot \Phi_{-\rho_1}(t_1, -t_1) + p_2 \cdot \Phi_{-\rho_2}(t_2, -t_1) + p_1 \cdot \Phi_{-\rho_1}(t_2, -t_2)\;,
\]
where $\rho_1, \rho_2 < 0$ are computed in Proposition~\ref{prop:positive}. We first look at the case where $-\infty<t_1,t_2<\infty$. The case where $t_1=\pm\infty$ or $t_2=\pm\infty$, which corresponds to always setting one or both variables to $1$ or $-1$, can be dealt with separately via a simple case analysis. 

As we discussed in the introduction, a $\THRESH^-$ rounding scheme for \MAXDICUT\ is not necessarily odd, but as the following lemma shows, the simple and symmetric structure of our construction ensures that any finite critical point of $s$ is necessarily symmetric around the origin.

\begin{lemma}
Let $x, y \in \mathbb{R}$. If $(x, y)$ is a critical point of $s(t_1, t_2)$, then $y = -x = |x|$.
\end{lemma}
\begin{proof}
Recall that by Proposition~\ref{prop:Phi_partial}
$\frac{\partial}{\partial t_1}\Phi_\rho(t_1, t_2) = \varphi(t_1) \cdot \Phi\left(\frac{t_2 - \rho t_1}{\sqrt{1 - \rho^2}}\right)$.

The partial derivatives of $s(t_1, t_2)$ are
\begin{align*}
\frac{\partial s}{\partial t_1} & = p_1 \left(\varphi(t_1) - 2\varphi(t_1)\cdot \Phi\left(\sqrt{\frac{1-\rho_1}{1+\rho_1}}t_1\right)\right) + p_2\left( - \varphi(t_1)\cdot\Phi\left(\frac{t_2 - \rho_2 t_1}{\sqrt{1-\rho_2^2}}\right)\right),\\
\frac{\partial s}{\partial t_2} & = p_1 \left(\varphi(t_2) - 2\varphi(t_2)\cdot \Phi\left(\sqrt{\frac{1-\rho_1}{1+\rho_1}}t_2\right)\right) + p_2\left(\varphi(t_2) - \varphi(t_2)\cdot\Phi\left(\frac{t_1 - \rho_2 t_2}{\sqrt{1-\rho_2^2}}\right)\right).
\end{align*}

In the above computation, we used Proposition~\ref{prop:Phi_partial} and the chain rule. Since $(x, y)$ is a critical point of $s$ and $\varphi$ is strictly positive, we have
\begin{align*}
 p_1 \left(1 - 2\Phi\left(\sqrt{\frac{1-\rho_1}{1+\rho_1}}x\right)\right) + p_2\left( - \Phi\left(\frac{y - \rho_2 x}{\sqrt{1-\rho_2^2}}\right)\right) & \EQ 0\;,\\
 p_1 \left(1 - 2\Phi\left(\sqrt{\frac{1-\rho_1}{1+\rho_1}}y\right)\right) + p_2\left(1 - \Phi\left(\frac{x - \rho_2 y}{\sqrt{1-\rho_2^2}}\right)\right) & \EQ 0\;.
\end{align*}

The first equation can be rewritten as
\begin{equation}\label{eq:1}
p_1 \left(1 - 2\Phi\left(\sqrt{\frac{1-\rho_1}{1+\rho_1}}x\right)\right) \EQ p_2\cdot  \Phi\left(\frac{y - \rho_2 x}{\sqrt{1-\rho_2^2}}\right)\;.
\end{equation}
Since $\Phi$ is a positive function, the right hand side of ($\ref{eq:1}$) is positive and therefore we have $1 - 2\Phi\left(\sqrt{\frac{1-\rho_1}{1+\rho_1}}x\right) > 0$, which implies that $x < 0$.

Since $1 - \Phi(t) = \Phi(-t)$, the second equation can be rewritten as
\begin{equation}\label{eq:2}
 p_1 \left(1 - 2\Phi\left(\sqrt{\frac{1-\rho_1}{1+\rho_1}}y\right)\right) \EQ - p_2\cdot\Phi\left(\frac{-x + \rho_2 y}{\sqrt{1-\rho_2^2}}\right)\;.
\end{equation}
By similar logic we can deduce that $y > 0$. We now show that we must have $|x| = |y|$. Assume for the sake of contradiction that $|x| \neq |y|$. We have two cases:
\begin{itemize}
    \item $|x| > |y|$. It follows that
    \begin{align*}
    p_1 \cdot \left|1 - 2\Phi\left(\sqrt{\frac{1-\rho_1}{1+\rho_1}}x\right)\right| 
    &\EQ p_1 \cdot \left|\Phi\left(-\sqrt{\frac{1-\rho_1}{1+\rho_1}}x\right) - \Phi\left(\sqrt{\frac{1-\rho_1}{1+\rho_1}}x\right)\right| \\
    &\GT p_1 \cdot \left|\Phi\left(-\sqrt{\frac{1-\rho_1}{1+\rho_1}}y\right) - \Phi\left(\sqrt{\frac{1-\rho_1}{1+\rho_1}}y\right)\right| \\
    &\EQ p_1 \cdot \left|1 - 2\Phi\left(\sqrt{\frac{1-\rho_1}{1+\rho_1}}y\right)\right| \;.  \\
    \end{align*}
    Note that here we again used  $1 - \Phi(t) = \Phi(-t)$, as well as the fact that $|\Phi(t) - \Phi(-t)|$ is an increasing function in $|t|$. On the other hand, by~(\ref{eq:1}) and (\ref{eq:2}) this implies that 
    \[
    \left|p_2\cdot  \Phi\left(\frac{y - \rho_2 x}{\sqrt{1-\rho_2^2}}\right)\right| \GT \left|- p_2\cdot\Phi\left(\frac{-x + \rho_2 y}{\sqrt{1-\rho_2^2}}\right)\right|\;.
    \]
    Since $\Phi$ is a positive and monotone function, this implies that 
    \[
    \frac{y - \rho_2 x}{\sqrt{1-\rho_2^2}} \GT \frac{-x + \rho_2 y}{\sqrt{1-\rho_2^2}}\;,
    \]
    Rearranging the terms, we obtain
    \[
    (1 - \rho_2)y \GT (1 - \rho_2) \cdot (-x)\;.
    \]
    But this would imply that $|y| > |x|$, which contradicts our assumption.
    \item $|y| > |x|$. This can be dealt with in a similar manner.
\end{itemize}
We conclude that we must have $y = -x = |x|$.
\end{proof}

\begin{lemma}\label{lem:unique_critical}
If $p_1 > p_2$, then $s(t_1, t_2)$ has a unique critical point. 
\end{lemma}
\begin{proof}
Assume $(x, y)$ is a critical point. In the previous lemma, we established that $x = -y < 0$, so we can now plug $y = -x$ into (\ref{eq:1}) and get
\[
p_1 \left(1 - 2\Phi\left(\sqrt{\frac{1-\rho_1}{1+\rho_1}}\cdot x\right)\right) \EQ p_2\cdot  \Phi\left(\frac{-1 - \rho_2 }{\sqrt{1-\rho_2^2}}\cdot x\right)\EQ p_2\cdot  \Phi\left(-\sqrt{\frac{1 + \rho_2 }{1 - \rho_2}}\cdot x\right)\;.
\]
We need to show the equation above has only one solution when $p_1 > p_2$. To this end, define 
\[
g(t) \EQ p_1 \left(1 - 2\Phi\left(\sqrt{\frac{1-\rho_1}{1+\rho_1}}\cdot t\right)\right) - p_2\cdot  \Phi\left(-\sqrt{\frac{1 + \rho_2 }{1 - \rho_2}}\cdot t\right)\;, \qquad t < 0\;.
\]

We have $g(0) = -p_2 < 0$ and $\lim_{t \to -\infty} g(t) = p_1 - p_2 > 0$, so $g(t) = 0$ has at least one solution in $(-\infty, 0)$ by Intermediate Value Theorem. To show that the solution is unique, we compute the derivative of $g$:
\begin{align*}
g'(t) &  \EQ p_1 \left( - 2 \sqrt{\frac{1-\rho_1}{1+\rho_1}} \cdot \varphi\left(\sqrt{\frac{1-\rho_1}{1+\rho_1}}\cdot t\right)\right) + p_2\cdot \sqrt{\frac{1 + \rho_2 }{1 - \rho_2}}\cdot\varphi\left(-\sqrt{\frac{1 + \rho_2 }{1 - \rho_2}}\cdot t\right)\;. \\
\end{align*}
By setting $g'(t) = 0$, we obtain 
\[
 2p_1 \sqrt{\frac{1-\rho_1}{1+\rho_1}} \cdot \varphi\left(\sqrt{\frac{1-\rho_1}{1+\rho_1}}\cdot t\right) = p_2\cdot \sqrt{\frac{1 + \rho_2 }{1 - \rho_2}}\cdot\varphi\left(-\sqrt{\frac{1 + \rho_2 }{1 - \rho_2}}\cdot t\right)
\]
Plugging in the definition of $\varphi$, we get
\[
 2p_1 \sqrt{\frac{1-\rho_1}{1+\rho_1}} \cdot \frac{1}{\sqrt{2\pi}}\exp\left(-\frac{1-\rho_1}{1+\rho_1}\cdot \frac{t^2}{2}\right) = p_2\cdot \sqrt{\frac{1 + \rho_2 }{1 - \rho_2}}\cdot\frac{1}{\sqrt{2\pi}}\exp\left(-\frac{1 + \rho_2 }{1 - \rho_2}\cdot \frac{t^2}{2}\right),
\]
which is equivalent to
\[
 2p_1 \sqrt{\frac{1-\rho_1}{1+\rho_1}} \cdot \exp\left(-\left(\frac{1-\rho_1}{1+\rho_1} - \frac{1 + \rho_2 }{1 - \rho_2}\right)\cdot \frac{t^2}{2}\right) \EQ p_2\cdot \sqrt{\frac{1 + \rho_2 }{1 - \rho_2}}\;.
\]
Since $\rho_1, \rho_2 < 0$ and $\exp$ is monotone, this equation has exactly one solution $t^* \in (-\infty, 0)$. Furthermore, $g'(t) > 0$ for $t \in (-\infty, t^*)$ and $g'(t) < 0$ for $t \in (t^*, 0)$. It follows that $g$ has no root in $(-\infty, t^*)$ and has a unique root in $(t^*, 0)$.
\end{proof}

We now deal with the boundary cases. Since our distribution is symmetric with respect to $\flip$, it is sufficient to look at the case where $t_1 = \pm \infty$.

\begin{lemma}\label{lemma:boundary}
We have $s(+\infty, +\infty) = s(-\infty, -\infty) = s(+\infty, -\infty) = 0$, $s(-\infty, +\infty) = p_2$. For $t_2 \in \mathbb{R}$, we have $s(-\infty, t_2) > s(+\infty, t_2)$. Furthermore, if $p_1 > p_2$, then $s(-\infty, t_2)$ is maximized when $t_2 = t^* = \sqrt{\frac{1 + \rho_1}{1 - \rho_1}}\cdot\Phi^{-1}(\frac{p_1 + p_2}{2p_1}).$
\end{lemma}
\begin{proof}
Setting a threshold to $+\infty$ corresponds to always setting a variable to false, and $-\infty$ corresponds to always true. When $(t_1, t_2) \in \{(+\infty, +\infty),(-\infty, -\infty),(+\infty, -\infty)\}$, none of the configurations are satisfied, giving a soundness of 0. When $(t_1, t_2) = (-\infty, +\infty)$, only the second configuration is satisfied and this gives a soundness of $p_2$.
For aim, we have 
\[s(-\infty, t_2) \EQ p_2 \cdot \Phi(t_2) + p_1\cdot\Phi_{-\rho_1}(t_2, -t_2)>p_1\cdot\Phi_{-\rho_1}(t_2, -t_2)\EQ s(+\infty, t_2)\;,\] and
\[
\frac{\partial s(-\infty, t_2)}{\partial t_2} \EQ \varphi(t_2) \left(p_2 + p_1\left(1  - 2\Phi\left(\sqrt{\frac{1-\rho_1}{1+\rho_1}}t_2\right)\right)\right)\;.
\]
When $p_1 > p_2$, we have $\frac{\partial s(-\infty, t_2)}{\partial t_2} > 0$ on $(-\infty, t^*)$ and $\frac{\partial s(-\infty, t_2)}{\partial t_2} < 0$ on $(t^*, \infty)$.
\end{proof}

With Lemma~\ref{lem:unique_critical} and Lemma~\ref{lemma:boundary}, it becomes very easy to determine the maximum of $s$ by simply computing the unique critical point and comparing it with the boundary cases. It turns out that when $b=0.1757079776$, $c= -0.6876930116$, $p_1= 0.3770580295$, the unique critical point of $s(t_1, t_2)$ is at $(-t_0, t_0)$ where $t_0\simeq 0.1887837358$, which is also a global maximum whose value is about $0.8746024732$.  A plot of $s(t_1, t_2)/\comp(\Tilde{\Theta})$ with these parameters can be found in Figure~\ref{fig:upper_bound}. It follows that with these parameters, any $\THRESH^-$ rounding scheme achieves a ratio of at most $0.87461$. This can then be converted into Unique Games hardness via standard and well-known techniques, which we include in the appendix for completeness.

\begin{figure}
\begin{center}
\includegraphics[width=3in]{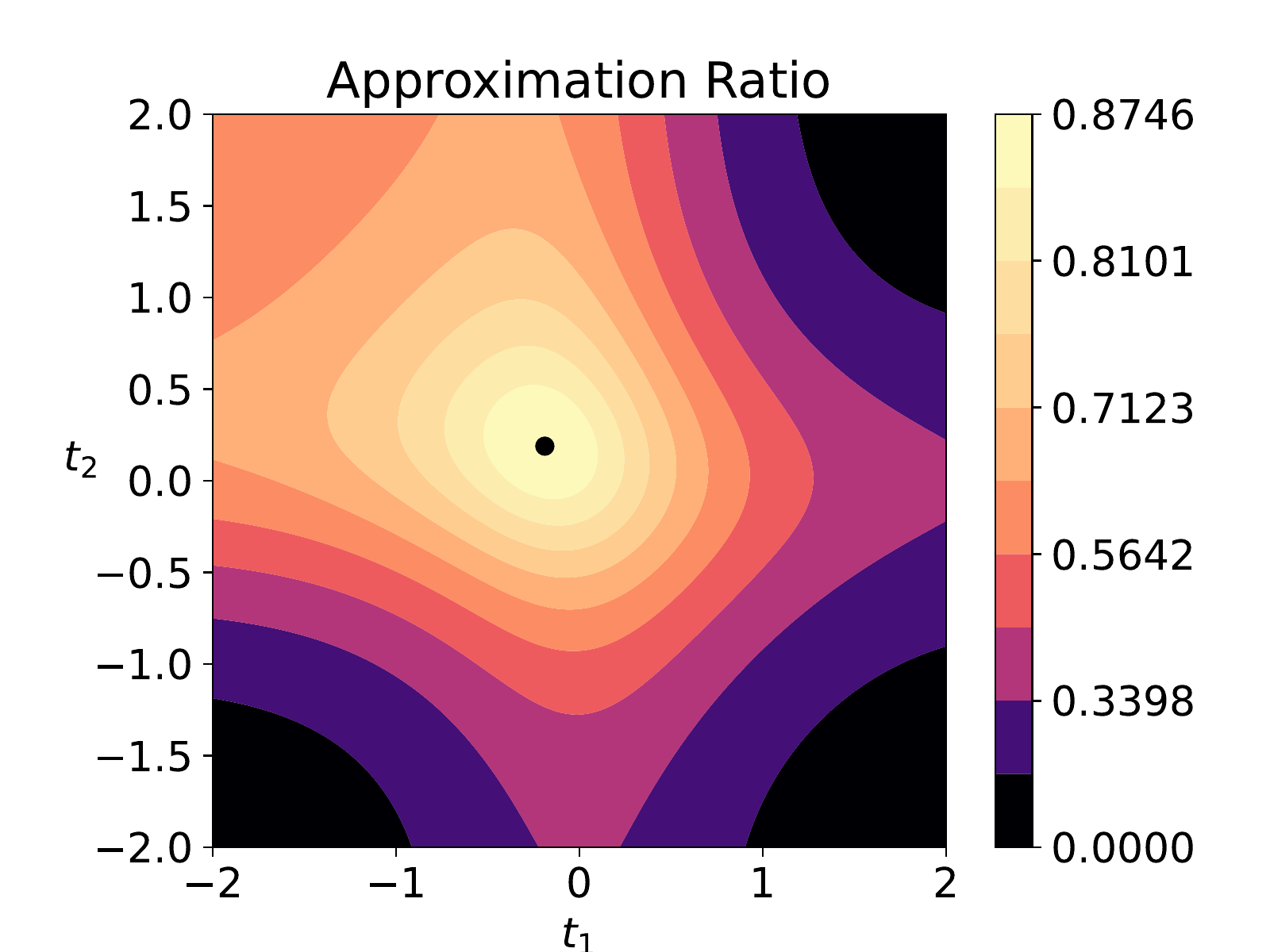}
\includegraphics[width=3in]{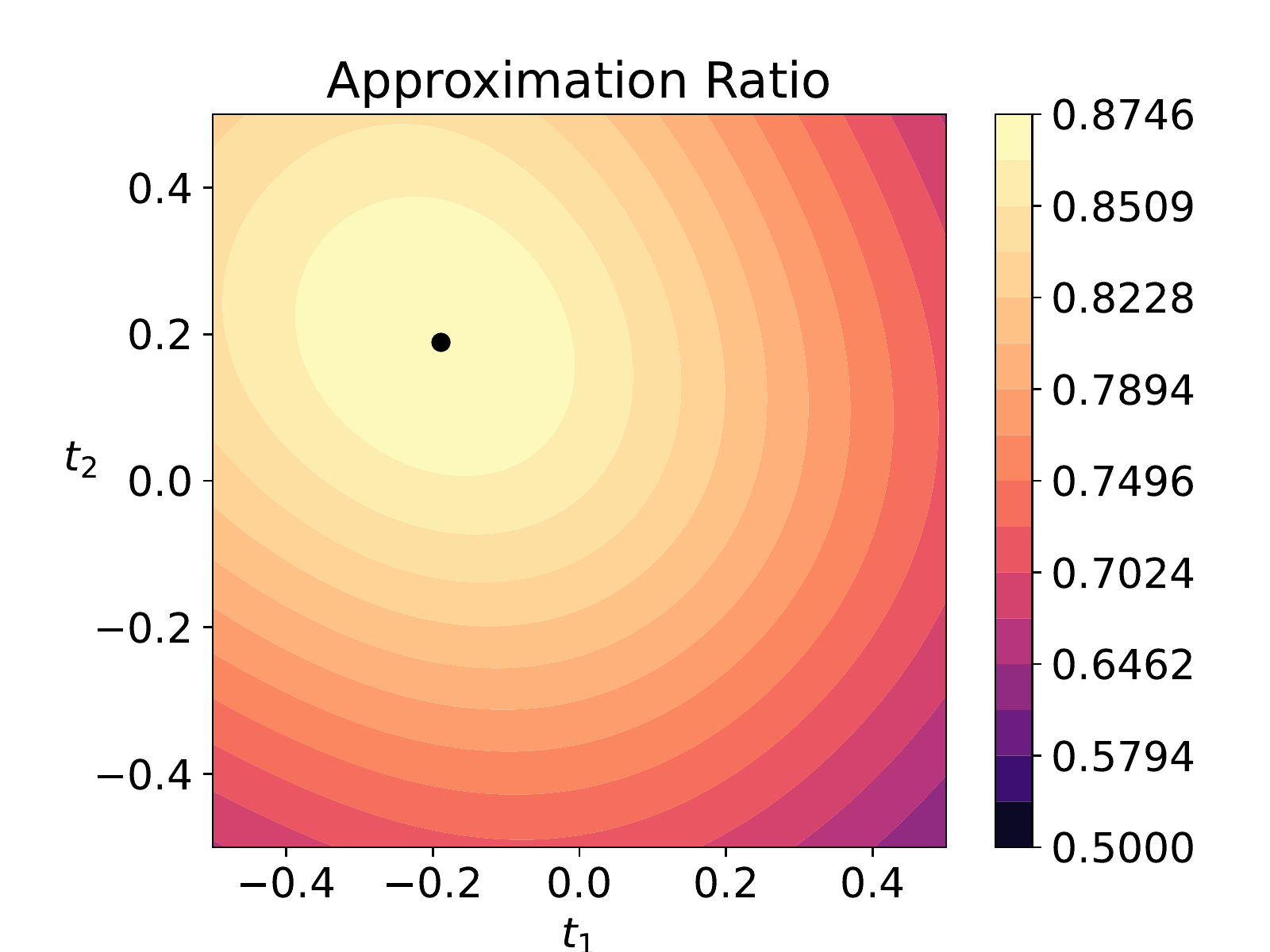}
\end{center}
\caption{Contour plots of $s(t_1, t_2)/\comp(\Tilde{\Theta})$ with optimal parameters. The black dot represents the global maximum $(-t_0, t_0)$ where $t_0 \simeq 0.1887837358$. All plots in this paper are made with Matplotlib~\cite{Hunter:2007}. %
}\label{fig:upper_bound}
\end{figure}

\subsection{Intuition for the upper bound}

While we found this integrality gap instance with a computer search, we now give some intuition for why this integrality gap instance works well. Previously, the best algorithm for \MAXDICUT\ was the LLZ algorithm \cite{LLZ02} which works equally well for \MAXAND.

If we restrict our attention to points $(b_1,b_2,-1 + |b_1 + b_2|)$ where the triangle inequality is tight (so the completeness is as large as possible given $b_1$ and $b_2$), using experimental simulations, the performance of LLZ in terms of $b_1$ and $b_2$ is as follows:

\begin{figure}[ht]
\begin{center}
\includegraphics[width=3in]{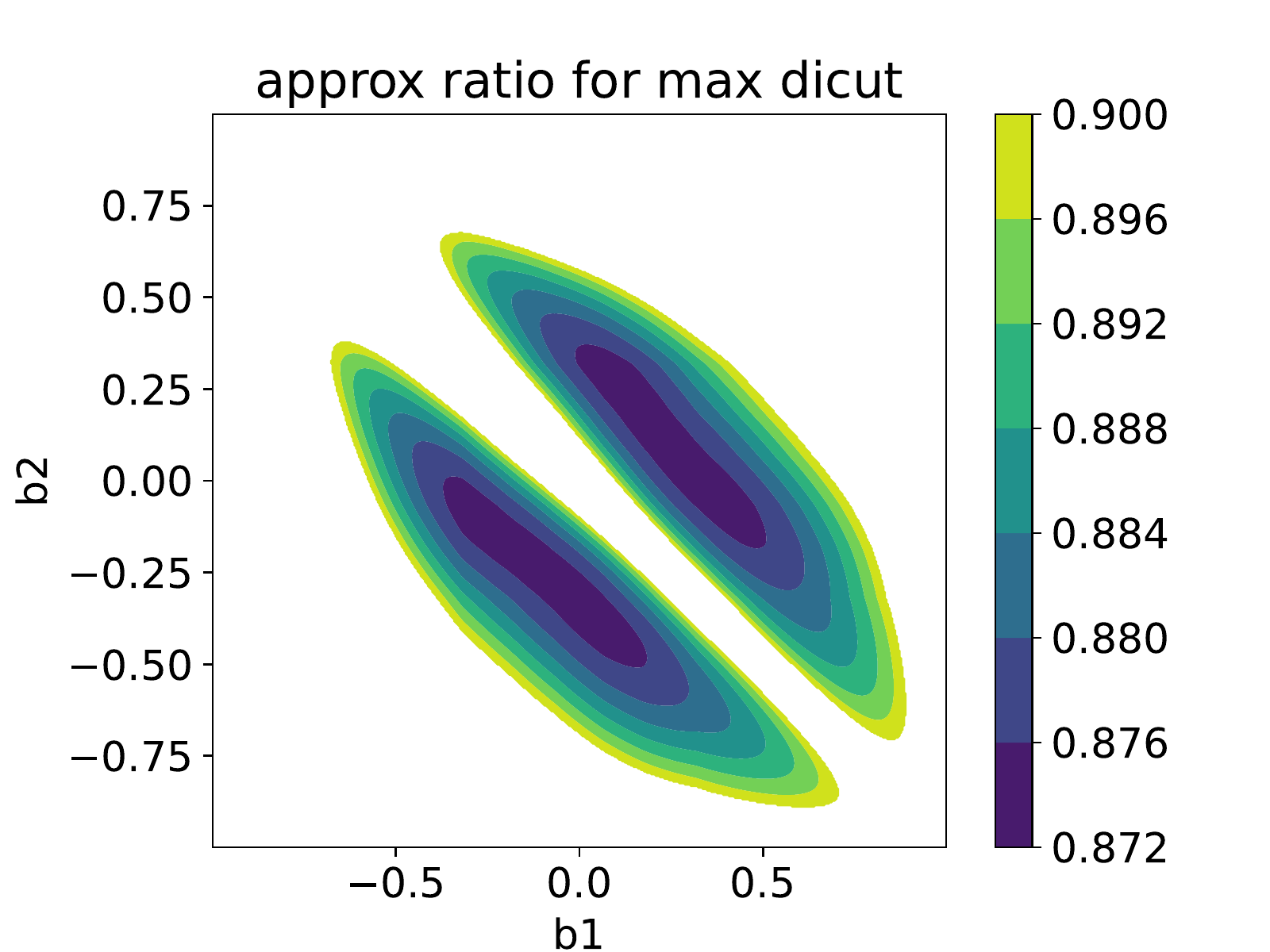}
\includegraphics[width=3in]{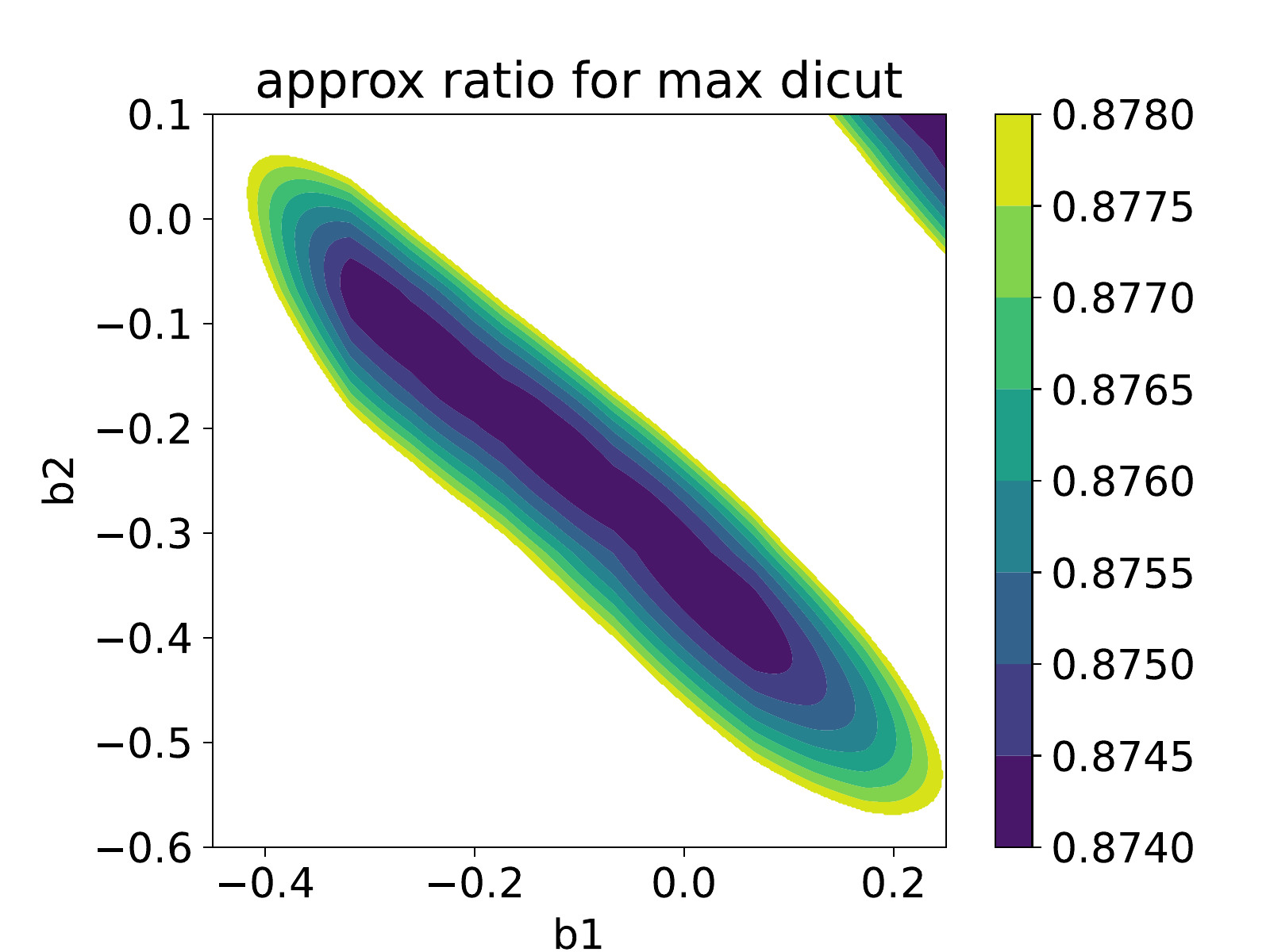}
\end{center}
\caption{Contour plots of the performance of the LLZ function~\cite{LLZ02} for \MAXAND\ and \MAXDICUT. %
}\label{fig:LLZ}
\end{figure}

We observe that there is a strip where $b_1 + b_2 \approx .35$ and a strip where $b_1 + b_2 \approx -.35$ where the LLZ algorithm does poorly. In order to reduce the degrees of freedom for rounding schemes for our instance, it makes sense to choose $b_2 = b_1 = \pm{b}$. With this choice, there are only two degrees of freedom, the threshold for $b$ and the threshold for $-b$. $b \simeq 0.1757079776$ puts us right in the middle of the hard strips for LLZ.

Once we have these two points, we can also add points of the form $(b,-b,c)$ and $(-b,b,c')$ without additional degrees of freedom. While we originally thought that points where the triangle inequality is tight may be optimal, this turned out to not be the case. Instead, we found experimentally that adding the point $(b,-b,c)$ with $c \simeq -0.6876930116$ worked best. The completeness for $(-b,b,c')$ is too low, so adding this kind of point does not help.

\subsection{Possibly improved upper bounds}

We believe that slightly improved upper bounds for \MAXDICUT\ can be obtained using more than one pair of biases. In Appendix~\ref{A-Upper} we give more complicated distributions that use up to 4 pairs of biases that seem to indicate that $\aDC\le 0.8745794663$ (not verified rigorously). It would probably be very hard to prove this inequality analytically. It is probably possible to prove that, say, $\aDC\le 0.8745795$, using interval arithmetic, but we have not done so yet.

\section{A new approximation algorithm for \texorpdfstring{\MAXDICUT}{MAX DI-CUT}}\label{sec:lower}

In this section, we present the techniques used for proving Theorem~\ref{theorem:lower}. We first briefly give some intuition for why a rounding scheme for \MAXDICUT\ better than those possible for \MAXAND\ should exist. Then, after describing the rounding scheme, we explain how this rounding scheme was discovered experimentally. Finally, we discuss how we rigorously verify the approximation guarantees of this rounding scheme using interval arithmetic.

\subsection{Intuition for the separation between \texorpdfstring{\MAXAND}{MAX 2-AND} and \texorpdfstring{\MAXDICUT}{MAX DI-CUT}}
We now try to give some intuition for why there is a gap between \MAXAND\ and \MAXDICUT. We first observe that Austrin's hard distributions of configurations for \MAXAND\ (see Section 6 of \cite{Austrin10}) can be easily beaten for \MAXDICUT. For simplicity, we consider Austrin's simpler two-configuration distribution which is as follows
\begin{enumerate}
    \item $(0,-b,b-1)$ with probability $0.64612$
    \item $(0,b,b-1)$ with probability $0.35388$ 
\end{enumerate}
where $b = 0.33633$. This gives an inapproximability of 0.87451 for \MAXAND.

For \MAXAND, since variables can be negated, we can assume without loss of generality that when $b = 0$, for each rounding scheme in our distribution the variable has an equal probability of being rounded to true or false.

For \MAXDICUT, we only have the symmetry of $\theta \mapsto \flip(\theta)$. When we add this symmetry to the integrality gap instance, we obtain:
\begin{enumerate}
    \item $(0,-b,b-1)$ with probability $0.32306$
    \item $(b,0,b-1)$ with probability $0.32306$
    \item $(0,b,b-1)$ with probability $0.17694$ 
    \item $(-b,0,b-1)$ with probability $0.17694$ 
\end{enumerate}
where $b = 0.33633$.

The following distribution of rounding functions trivially satisfies $\frac{1}{2}$ of the configurations of this \MAXDICUT\  instance.
\begin{enumerate}
    \item With probability $\frac{1}{2}$, round all variables with bias $0$ to $1$ and round all variables with bias~$-b$ or $b$ to $-1$.
    \item With probability $\frac{1}{2}$, round all variables with bias $0$ to $-1$ and round all variables with bias~$-b$ or $b$ to $1$.
\end{enumerate}
Since the completeness of these configurations are all at most $\frac{1}{2}$, we obtain a ratio which is at least~$1$.

While this is an extreme example, this shows that making the variables with zero or low bias more likely to be rounded to $1$ or more likely to be rounded to $-1$ can help round other variables more effectively as the behavior of these variables is more predictable.

This means that configurations where one or more variables has bias $0$ are easier for \MAXDICUT. Instead, the configurations in our simple distribution (i.e., $(b,b,-1 + 2b)$ where $b\simeq0.1757079776$) are hard configurations and all rounding schemes in the distribution have essentially the same behavior at these configurations. %

\subsection{The rounding scheme}\label{sec:scheme}

We now describe a $\THRESH$ scheme, that separates \MAXDICUT\ from \MAXAND. As mentioned, a $\THRESH$ scheme is a distribution over $\THRESH^-$ schemes. We use discrete distributions over a relatively small number of $\THRESH^-$ schemes. For computational convenience, we choose the $\THRESH^{-}$ functions to be \emph{piecewise linear} functions. More precisely, we pick a finite set $S \subset [-1, 1]$ of \emph{control points} with $-1, 1 \in S$. For each of these control points $s \in R$, we assign a real threshold $f(s)$. Then, for every $x \in (-1, 1) \setminus S$, we identify $x_{-} = \max (S \cap [-1, x))$ and $x_{+} = \min (S \cap (x, 1])$, and set
\[
    f(x) \EQ f(x_{-}) + \frac{x - x_{-}}{x_{+} - x_{-}} (f(x_{+}) - f(x_{-})) \;.
\]
We use the same set of control points for every function in our $\THRESH$ scheme.
 
For our application to \MAXDICUT, we picked a set $S$ of $17$ control points as follows:
\[ 0 \quad \pm 0.1 \quad \pm 0.164720 \quad \pm 0.179515 \quad \pm 0.25 \quad \pm 0.3 \quad \pm 0.45 \quad \pm 0.7 \quad \pm 1 \]
The choice of most control points is fairly arbitrary. It seemed important, however, to choose the four control points $\pm 0.164720$ and $\pm 0.179515$ as they seem to be situated in regions in which very fine control over the values of the rounding functions are needed. Further small improvement are probably possible by slightly moving some of the control points or by adding new control points.

Then, using the algorithm presented in Section~\ref{sec:discovery}, we produced a ``raw'' $\THRESH$ rounding scheme which is a probability distribution over $39$ piecewise-linear rounding functions. %
After a careful ad-hoc analysis, we were able to simplify the distribution to a ``clean'' $\THRESH$ scheme with only $7$ piecewise rounding functions, which we summarize in Table~\ref{tbl:clean-dicut} and Figure~\ref{fig:clean-dicut}. 

It is interesting to note that the function $f_1$, which is used in about $99.7\%$ of the time, is very close to the single function used by~\cite{LLZ02}. We do not yet a satisfactory explanation of the shape of the other functions. Some of the values of the functions, especially at control points $\pm 0.45$, $\pm 0.7$ and $\pm 1$ can be changed slightly without affecting the performance ratio obtained. We also note that the last two functions do not seem to contribute much. We have a scheme with only 5 functions with only a very slightly smaller performance ratio. 

\begin{table}
\begin{center}
\scriptsize
\begin{tabular}{r|r|rr|rr|rr}
 &  \hfill\nobreak $f_1$  \hfill\nobreak  &  \hfill\nobreak $f_2$  \hfill\nobreak  &  \hfill\nobreak $f_3$  \hfill\nobreak  &  \hfill\nobreak $f_4$  \hfill\nobreak  &  \hfill\nobreak $f_5$  \hfill\nobreak  &  \hfill\nobreak $f_6$  \hfill\nobreak  &  \hfill\nobreak $f_7$  \hfill\nobreak  \\
\hline
\hfill\nobreak prob \hfill\nobreak  & $ 0.996902$ & $ 0.000956$ & $ 0.000956$ & $ 0.000393$ & $ 0.000393$ & $ 0.000200$ & $ 0.000200$\\
\hline
$-1.000000$ & $-1.601709$ & $-2.000000$ & $-2.000000$ & $-0.034381$ & $-0.430994$ & $-2.000000$ & $ 2.000000$ \\
$-0.700000$ & $-0.853605$ & $-2.000000$ & $-2.000000$ & $-0.034381$ & $-0.430994$ & $-2.000000$ & $ 2.000000$ \\
$-0.450000$ & $-0.517014$ & $-2.000000$ & $-0.629564$ & $-0.440988$ & $-0.896878$ & $-2.000000$ & $ 2.000000$ \\
$-0.300000$ & $-0.333109$ & $-1.520523$ & $ 1.711824$ & $-1.406591$ & $ 1.643936$ & $-2.070000$ & $ 1.970000$ \\
$-0.250000$ & $-0.274589$ & $-0.687582$ & $ 2.019266$ & $-0.622399$ & $-0.127984$ & $-1.629055$ & $ 2.070000$ \\
$-0.179515$ & $-0.192926$ & $-0.195474$ & $-0.229007$ & $-0.268471$ & $-0.339566$ & $-0.544957$ & $-0.103307$ \\
$-0.164720$ & $-0.175942$ & $-0.381789$ & $-0.649998$ & $-0.116530$ & $-0.073069$ & $-0.361234$ & $-0.575047$ \\
$-0.100000$ & $-0.105428$ & $-0.026636$ & $-1.175439$ & $ 0.066139$ & $-0.123693$ & $ 2.070000$ & $-1.351740$ \\
\hline
$ 0.000000$ & $ 0.000000$ & $ 2.046025$ & $-2.046025$ & $ 1.728858$ & $-1.728858$ & $ 2.050000$ & $-2.050000$ \\
\hline
$ 0.100000$ & $ 0.105428$ & $ 1.175439$ & $ 0.026636$ & $ 0.123693$ & $-0.066139$ & $ 1.351740$ & $-2.070000$ \\
$ 0.164720$ & $ 0.175942$ & $ 0.649998$ & $ 0.381789$ & $ 0.073069$ & $ 0.116530$ & $ 0.575047$ & $ 0.361234$ \\
$ 0.179515$ & $ 0.192926$ & $ 0.229007$ & $ 0.195474$ & $ 0.339566$ & $ 0.268471$ & $ 0.103307$ & $ 0.544957$ \\
$ 0.250000$ & $ 0.274589$ & $-2.019266$ & $ 0.687582$ & $ 0.127984$ & $ 0.622399$ & $-2.070000$ & $ 1.629055$ \\
$ 0.300000$ & $ 0.333109$ & $-1.711824$ & $ 1.520523$ & $-1.643936$ & $ 1.406591$ & $-1.970000$ & $ 2.070000$ \\
$ 0.450000$ & $ 0.517014$ & $ 0.629564$ & $ 2.000000$ & $ 0.896878$ & $ 0.440988$ & $-2.000000$ & $ 2.000000$ \\
$ 0.700000$ & $ 0.853605$ & $ 2.000000$ & $ 2.000000$ & $ 0.430994$ & $ 0.034381$ & $-2.000000$ & $ 2.000000$ \\
$ 1.000000$ & $ 1.601709$ & $ 2.000000$ & $ 2.000000$ & $ 0.430994$ & $ 0.034381$ & $-2.000000$ & $ 2.000000$ \\
\end{tabular}
\end{center}

\caption{A $\THRESH$ rounding scheme that gives a rigorously verified approximation ratio of at least $0.874473$ for \MAXDICUT. (The actual ratio is probably about $0.874502$.) The scheme uses 7 piecewise-linear rounding functions $f_1,f_2,\ldots,f_7$ defined on 17 control points. The function $f_1$ is odd and is very close to the single function used by \cite{LLZ02}. The other six functions come in pairs. The two functions in each pair are flips of each other.}\label{tbl:clean-dicut}
\end{table}

\begin{figure}[t]

\begin{center}
\begin{tabular}{cc}
\includegraphics[width=3in]{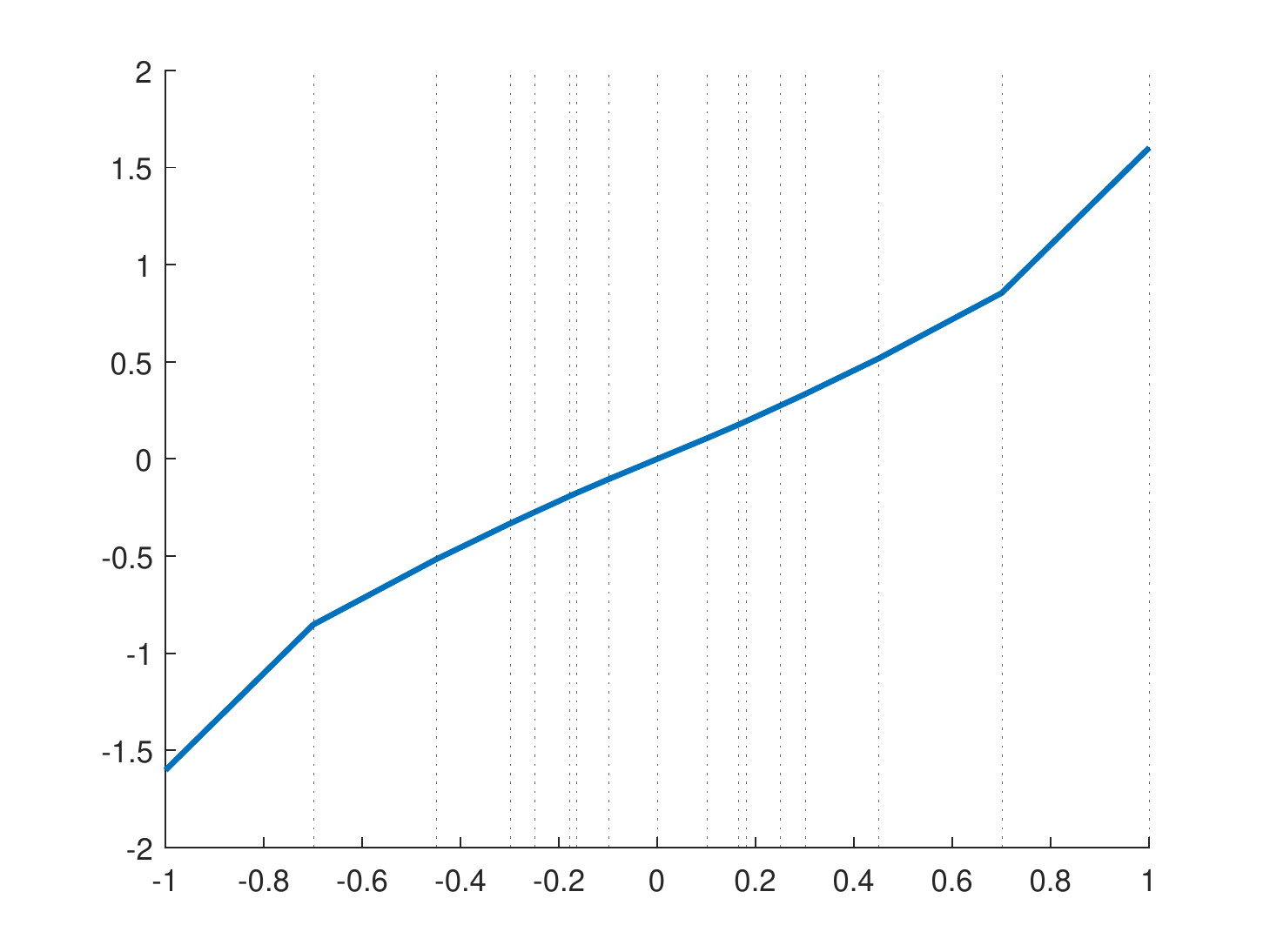} &
\includegraphics[width=3in]{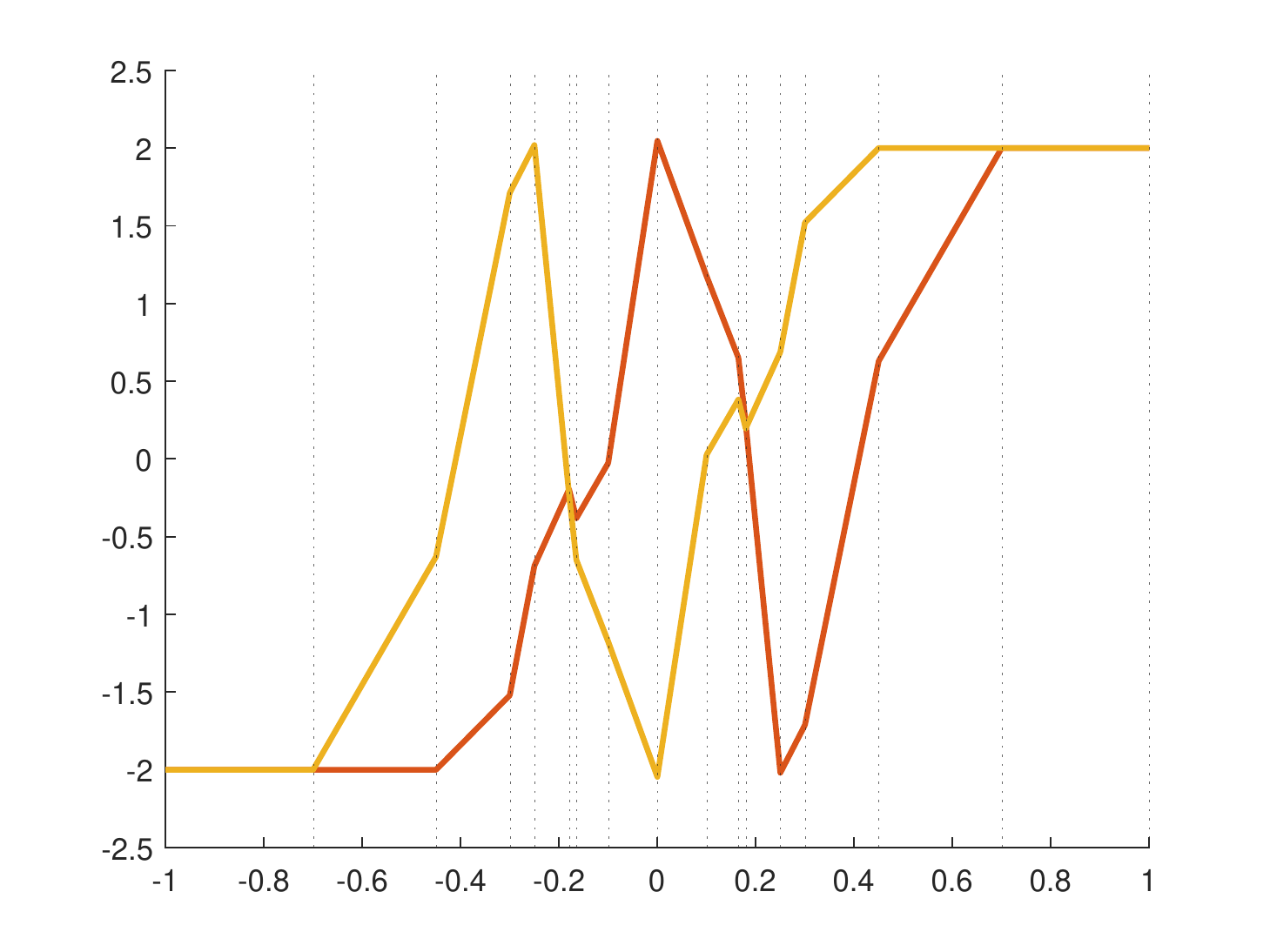} \\[-7pt]
$f_1$ with probabilities $0.996902$  & $f_2$ and $f_3$ each with probability $0.000956$ \\[5pt]
\includegraphics[width=3in]{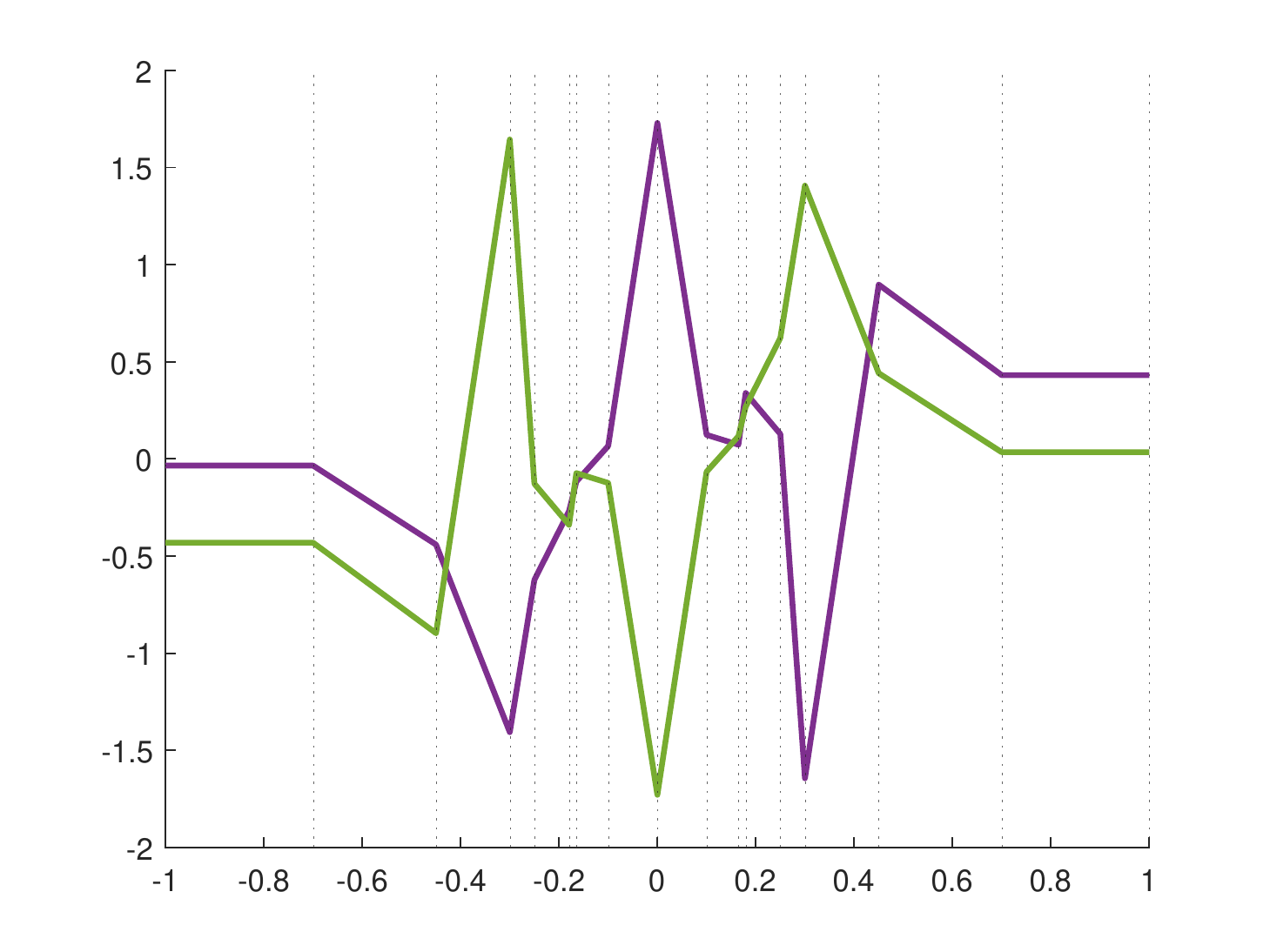} &
\includegraphics[width=3in]{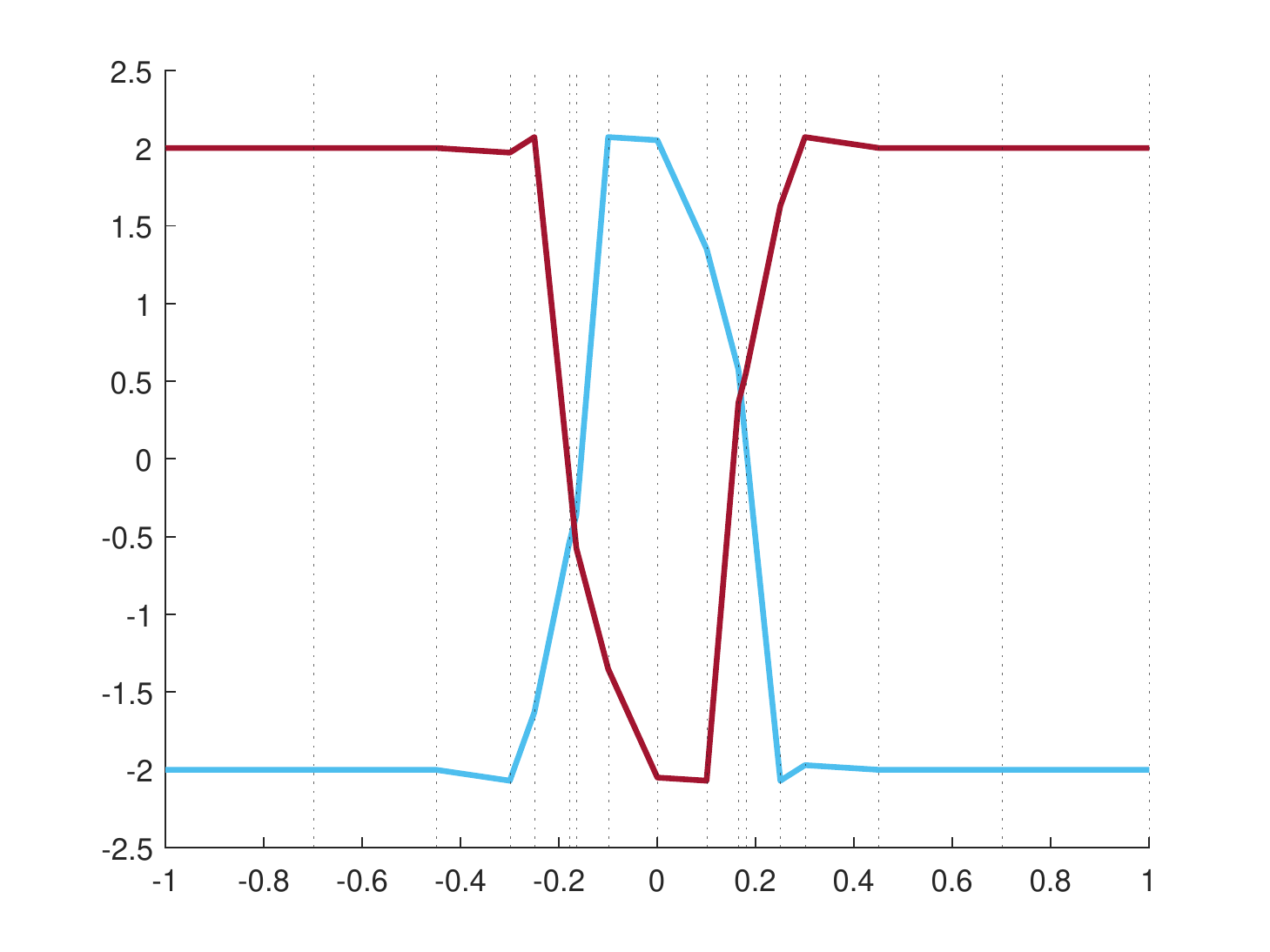} \\[-7pt]
$f_4$ and $f_5$ each with probability $0.000393$ & $f_6$ and $f_7$ each with probability $0.000200$
\end{tabular}
\end{center}
\caption{Plots of the seven rounding functions used in the $\THRESH$ rounding scheme given in Table~\ref{tbl:clean-dicut} that achieves a verified approximation ration of at least $\bestDICUT$ for \MAXDICUT.}\label{fig:clean-dicut}
\end{figure}

\begin{figure}[t]

\begin{center}
\includegraphics[width=5in]{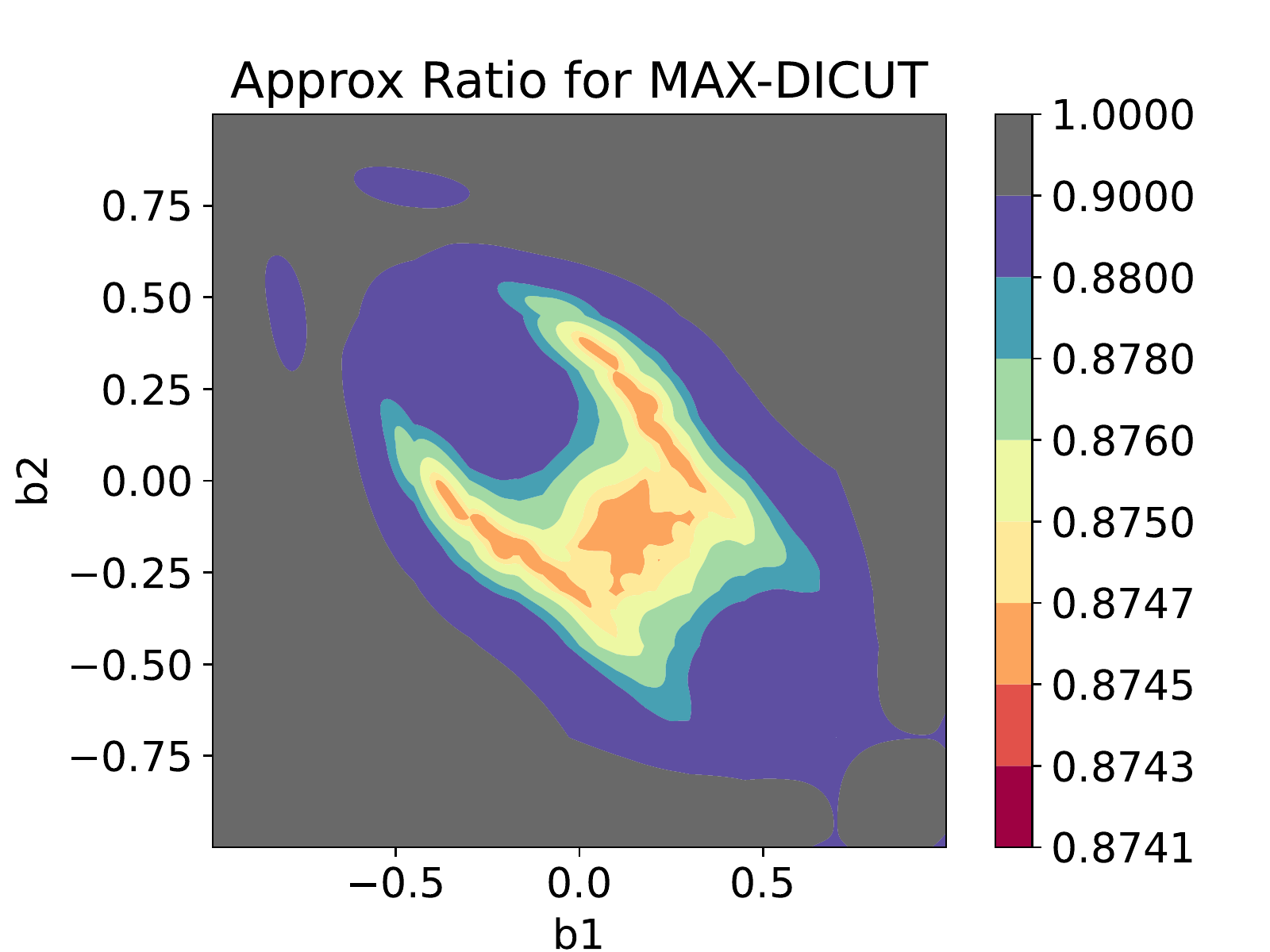}
\end{center}

\caption{This plot is a contour plot of the performance of the $\THRESH$ scheme with $7$ piecewise-linear rounding functions for various choices of $b_1$ and $b_2$ (with an approximately worst-case choice of $b_{12}$) selected.}
\end{figure}

\subsection{Discovery of the \texorpdfstring{$\THRESH$}{THRESH} scheme}\label{sec:discovery}

In this section, we discuss the process of experimentally discovering the ``raw'' $\THRESH$ scheme described in Section~\ref{sec:scheme}. For now, we will make a couple of assumptions, which will be fully worked out in Section~\ref{sub:details}.

\begin{itemize}
\item[(1)] Instead of optimizing over all valid configurations of \MAXDICUT, we restrict to optimizing over a finite set $\Theta$ of configurations, where $\comp(\theta) > 0$ for all $\theta \in \Theta$.
\item[(2)] Let $\mathcal F$ be a restricted family of $\THRESH^{-}$ schemes (e.g., the piecewise linear functions).  We shall further assume throughout this discussion that we have access to an oracle $\mathcal O_{\mathcal F}$ which, when given a probability distribution $\tilde{\Theta} \in \Dist(\Theta)$, identifies a function $f \in \mathcal F$ which maximizes $\sound(\tilde{\Theta}, f)$. 
\end{itemize}

\subsubsection{Finite \texorpdfstring{$F$}{F}: a game-theoretic approach}
Assume further we have found a finite set $F \subset \mathcal F$ of candidate rounding functions. We would like to identify the following:

\begin{itemize}
\item[(a)] An optimal (worst) distribution $\tilde{\Theta}$ over $\Theta$ such that
\[\tilde{\Theta} \EQ \underset{\tilde{\Theta} \in \Dist(\Theta)}{\argmin} \; \max_{f \in F} \; \frac{\sound(\tilde{\Theta}, f)}{\comp(\tilde{\Theta})}\;.\]
\item[(b)] An optimal (best) distribution $\tilde{F}$ over $F$ such that
\[
\tilde{F} \EQ \underset{\tilde{F} \in \Dist(F)}{\argmax} \; \min_{\theta \in \Theta} \; \underset{f \sim \tilde{F}} {\mathbb E}  \left[\frac{\sound(\tilde{\Theta}, f)}{\comp(\tilde{\Theta})}\right]\;.
\]
\end{itemize}

It turns out that both of these objectives can be solved by mutually dual LPs. This is best seen by casting both questions as a \emph{zero-sum game}. Fix a real number $\alpha$, which should be thought of as an estimate of the approximation ratio of this restricted \MAXDICUT\ problem. In our game, which we call \emph{the $\alpha$-game}, there are two players Alice and Bob that play simultaneously: Alice picks $\theta \in \Theta$ and Bob picks $f \in F$. We then have the following payoffs
\begin{align*}
\text{Alice: }\;&  \text{alice}_{\alpha}(\theta, f) \;:=\; \alpha\, \comp (\theta) - \sound(\theta, f)\\
\text{Bob: }\;& \text{bob}_{\alpha}(\theta, f) \;:=\; \sound(\theta, f) - \alpha\, \comp (\theta)
\end{align*}
Note that this game is a finite zero-sum game and thus by standard theory (e.g., Von Neumann's minimax theorem~\cite{neumann1928theorie} and Nash equilibria~\cite{nash1951non}), there is a single\footnote{Depending on the singular values of the payoff matrices, there may be multiple Nash-equilibrium, but they all have the same value. In that situation, we pick one of the Nash equilibriums arbitrarily to be representative Nash equilibrium.} Nash-equilibrium $(\tilde{\Theta}_{\alpha}, \tilde{F}_{\alpha})$ which is the optimal mixed strategy for both players. Let $v(\alpha)$ be the expected payoff of this optimal strategy for Alice (i.e., the \emph{value} of the game). We now make the following simple observation.

\begin{proposition}\label{prop:strict}
The function $v(\alpha)$ is strictly increasing in~$\alpha$. 
\end{proposition}
\begin{proof}
Fix $\alpha < \alpha'$. Assume for the $\alpha'$-game that Alice plays $\tilde{\Theta}_{\alpha}$. Assume Bob plays an arbitrary mixed strategy $\tilde{F}$. Then, Alice's expected payoff is
\begin{align*}
\underset{\theta \sim \tilde{\Theta}_{\alpha}, f \sim \tilde{F}}{\E}[\text{alice}_{\alpha'}(\theta, f)] &\EQ \underset{\theta \sim \tilde{\Theta}_{\alpha}, f \sim \tilde{F}}{\E}[\text{alice}_{\alpha}(\theta, f) + (\alpha'-\alpha)\comp (\theta)]\\
&\EQ \underset{\theta \sim \tilde{\Theta}_{\alpha}, f \sim \tilde{F}}{\E}[\text{alice}_{\alpha}(\theta, f)] + (\alpha'-\alpha)\underset{\theta \sim \tilde{\Theta}_{\alpha}}{\E}[\comp (\theta)]\\
&\GT v(\alpha)\;,
\end{align*}
where we use the fact that $\tilde{\Theta}_{\alpha}$ is the Nash equilibrium for the $\alpha$-game and that $\comp(\theta) > 0$ for all $\theta \in \Theta$. In other words, Alice can assure for the $\alpha'$-game a payoff strictly greater than $v(\alpha)$. Thus, $v(\alpha') > v(\alpha)$.
\end{proof}

It is easy to see that $v(0) \le 0$ (as $\text{alice}_{0} \le 0$). Further $\text{alice}_{\alpha}(\theta, f) \to \infty$ as $\alpha \to \infty$. Thus, by Proposition~\ref{prop:strict}, there is a unique $\alpha_{\Theta, F}$ for which $v(\alpha_{\Theta, F}) = 0$ with a corresponding Nash equilibrium of $\Tilde \Theta_{F}$ and $\Tilde F_{\Theta}$. Unpacking the definition of Nash equilibrium and using the fact that $\sound(\tilde{\Theta}, f)$ is an affine function in $\tilde{\Theta}$, we have that
\begin{itemize}
\item[(a)] For all $f \in F$, we have that \[\frac{\sound(\tilde{\Theta}_{F}, f)}{\comp(\tilde{\Theta})} \LE \alpha\;.\]
\item[(b)] For all $\theta \in \Theta$, we have that \[\underset{f \sim \tilde{F}_\Theta}{\mathbb E} \left[\frac{\sound(\theta, f)}{\comp(\theta)}\right] \GE \alpha\;.\]
\end{itemize}
Thus, $\Tilde \Theta_F$ and $\Tilde F_{\Theta}$ are the optimal distributions for problems (a) and (b) from before. We can efficiently compute these distributions through a suitable linear program. Let $w_\theta$ be the weights of the optimal distribution $\Tilde \Theta_F$ and let $p_f$ be the weights of the optimal distribution $\Tilde F_\Theta$. By definition of the Nash equilibrium, we have that
\begin{align}
& \sum_{\theta \in \Theta} w_{\theta} (\alpha_{\Theta, F}\, \comp (\theta) - \sound(\theta, f)) \GE 0 \quad,\quad \forall f \in F \label{eq:alice}\\
& \sum_{f \in F} p_f (\sound(\theta, f) - \alpha_{\Theta, F}\, \comp (\theta)) \GE 0 \quad,\quad \forall \theta \in \Theta \label{eq:bob}
\end{align}
To formulate this as a pair of linear programs, we will have $\alpha_{\Theta, F}$ be our objective. Since $v(\alpha) \ge 0$ for all $\alpha \ge \alpha_{\Theta, F}$ we will have a ``minimize'' objective to compute the $w_{\theta}$'s and a ``maximize'' objective to compute the $p_f$'s. 

However, neither set of constraints is currently an LP as $\alpha_{\Theta, F}$ is also a variable of our LP (in fact the objective function). This is easy to fix for (\ref{eq:bob}), as $\sum_{f \in F} p_f = 1$, so we can rewrite~(\ref{eq:bob}) as
\[
\sum_{f \in F} p_f \sound(\theta, f) \GE \alpha_{\Theta, F}\, \comp (\theta)\;.
\]
For (\ref{eq:alice}), we use a `clever' trick. We renormalize the weights so that $\sum_{\theta} \hat{w}_\theta\, \comp(\theta) = 1$ instead of $\sum_{\theta} w_\theta = 1$, and use the $\hat{w}_\theta$'s as the variables of the LP. With this normalization, we then get the linear constraints
\[
    \sum_{\theta \in \Theta} \hat{w}_\theta\, \sound(\theta, f) \LE \alpha_{\Theta, F}\;.
\]
After solving the LP, We can find the original weights by setting $w_\theta = \hat{w}_\theta / \sum_{\theta' \in \Theta} \hat{w}_{\theta'}$. Formally, the two LPs we solve are as follows.

\begin{equation*}
\boxed {
\begin{aligned}
& \quad \quad \text{Primal LP (finding ${\tilde \Theta}_F$)}\\
\textbf{minimize: } & \alpha\\
\textbf{subject to: } & \sum_{\theta \in \Theta} \sound(\theta_i, f) \hat{w}_\theta \le \alpha \quad,\quad \forall f \in F\\
& \sum_{\theta \in \Theta} \comp(\theta_i) \hat{w}_\theta = 1\\
& \hat{w}_\theta \ge 0 \quad,\quad \forall \theta \in \Theta
\end{aligned}
}
\end{equation*}

\begin{equation*}
\boxed {
\begin{aligned}
& \quad \quad \quad \quad \quad \quad  \text{Dual LP (finding ${\tilde F}_\Theta$)}\\
\textbf{maximize: } & \alpha\\
\textbf{subject to: } & \sum_{f \in F} \sound(\theta, f_j) p_f \ge \alpha\,\comp (\theta) \quad,\quad \forall \theta \in \Theta\\
& \sum_{f \in F} p_f = 1\\
& p_f \ge 0 \quad,\quad \forall f \in F
\end{aligned}
}
\end{equation*}

It is straightforward to prove that these two LPs are dual to each other and thus will both achieve the same objective value $\alpha_{\Theta,F}$

\subsubsection{Extending to infinite \texorpdfstring{$\mathcal F$}{F}}\label{sub:extend-F}

Since the full family of functions we optimize over is infinite, we cannot hope to find a (near) optimal distribution over $\mathcal F$ by just solving a suitable finite linear program. Instead, we work with a small set of candidate functions $F$ which we iteratively improve. In particular, for a fixed $F$, we can compute the hardest distribution $\tilde{\Theta}_F$ for this family of functions and then use the oracle to find the function $f = \mathcal O_{\mathcal F}(\tilde{\Theta}_F)$ which does the best on this hard distribution. (Note that $\mathcal O_{\mathcal F}(\tilde{\Theta}_F)$ needs to solve a non-linear, and probably non-convex, optimization problem.)  We add $f$ to $F$ and continue for some fixed number $T$ of steps. We note that similar minimax algorithms are prevalent in machine learning, such as in generative adversarial networks~\cite{goodfellow2020generative}. See Algorithm~\ref{alg:thresh} for the formal details.

\begin{algorithm}
\caption{$\THRESH$ discovery algorithm (fixed $\Theta$)}\label{alg:thresh}
\begin{algorithmic}[1]

\Procedure{FindThresh}{$\Theta$, $T$}

\State Pick an initial distribution $\Tilde \Theta_0 \in \Dist(\Theta)$
\State $f_1 \leftarrow \mathcal O_{\mathcal F}(\Tilde \Theta_0)$
\State $F_1 \leftarrow \{f_1\}$.
\For {$i \in \{1, 2, \hdots, T-1\}$}
    \State Find the hardest $\tilde \Theta_i $ for $F_{i}$ using the Primal LP with objective value $\alpha_i$
    \State $f_{i+1} \leftarrow \mathcal O_{\mathcal F}(\tilde \Theta_i)$
    \State $F_{i+1} \leftarrow F_{i} \cup \{f_{i+1}\}$
\EndFor

\State Find the optimal distribution $\tilde F_T$ over $F_T$ for $\Theta$ using the Dual LP

\State \Return $\tilde F_T$
  
\EndProcedure
\end{algorithmic}
\end{algorithm}

It is easy to see that the objective value $\alpha_i$ of the Primal LP in Algorithm~\ref{alg:thresh} increases at each step of the loop. Further,  it is not hard to prove that $\alpha_i \le 1/\min_{\theta \in \Theta} \comp(\theta)$.  Thus, as $T \to \infty$, the objective value of the Primal LP tends to a limit $\alpha_{\lim}$. Our main correctness guarantee of our algorithm is that we converge to this limit at an effective rate and that this limit is the best we can hope for.

\begin{theorem}\label{thm:alg-correct}
Fix $\eps > 0$ and assume $C := 1/\min_{\theta \in \Theta} \comp(\theta)$ is finite. Let $\alpha_{T}$ be the objective value of the Dual LP computing $\Tilde F_T$. Assume that $T > (C/\eps)^{|\Theta|}$, then $\alpha_T \ge \alpha_{\lim} - \eps$. Further, for every finite distribution $\Tilde F$ over functions in $\mathcal F$, 
\[\underset{f \sim \tilde{F}}{\mathbb E} \left[\frac{\sound(\tilde{\Theta}, f)}{\comp(\tilde{\Theta})}\right] \LE \alpha_{\lim}\;.\]
\end{theorem}

\begin{proof}

Observe that for all $j > i \ge 1$, we have that \[\frac{\sound(\Tilde \Theta_i, f_{i+1})}{\comp(\Tilde \Theta_i)} \GE \frac{\sound(\Tilde \Theta_j, f_{j})}{\comp(\Tilde \Theta_j)}\;,\]  because the distribution $\Tilde \Theta_i$ certifies that no finite distribution of rounding functions over $F_j$ can do better than $\frac{\sound(\Tilde \Theta_i, \mathcal O_{\mathcal F}(\Tilde \Theta_i))}{\comp(\Tilde \Theta_i)} $. In particular, by taking the limit as $j \to \infty$, this implies that for all $i \ge 1$,
\begin{align}
\frac{\sound(\Tilde \Theta_i, f_{i+1})}{\comp(\Tilde \Theta_i)} \GE \alpha_{\lim}\;.\label{eq:abc}
\end{align}

Assume for sake of contradiction that $\alpha_T < \alpha_{\lim} - \eps$. Thus, $\alpha_i < \alpha_{\lim} - \eps$ for all $i \in \{1, \hdots, T-1\}$. Pick $\delta = \eps / C$. Define the function $s_\Theta : \mathcal F \to [0, 1]^{\Theta}$ as
\[
    s_\Theta(f) \EQ (\sound(\theta, f) : \theta \in \Theta)\;.
\]
Observe that if $f$ and $f'$ are such that $\|s_\Theta(f) - s_\Theta(f')\|_{\infty} \le \delta$, then for any fixed distribution $\Tilde \Theta$, we have that 
\begin{align}
\left|\frac{\sound(\Tilde \Theta, f)}{\comp(\Tilde \Theta)} - \frac{\sound(\Tilde \Theta, f')}{\comp(\Tilde \Theta)}\right| \LE \frac{\delta}{\min_{\theta \in \Theta} \comp(\theta)} \EQ \eps\;.\label{eq:eps-net}
\end{align}
Divide $[0,1]^{\Theta}$ in $(1/\delta)^{|\Theta|}$ hypercubes with $\ell_\infty$-diameter $\delta$. Let $\mathcal H$ be the this family of hypercubes. Since $T > (C/\eps)^{|\Theta|}$, by the pigeonhole principle there exists $i,j \in \{1,\hdots, T\}$ with $f_i$ and $f_j$ in the same hypercube but $i < j$. In particular, we have by the minimax guarantee of $\Theta_{j-1}$, (\ref{eq:eps-net}), and (\ref{eq:abc}),
\begin{align}
\alpha_T \GE \alpha_{j-1} \GE \frac{\sound(\Tilde \Theta_{j-1}, f_{i})}{\comp(\Tilde \Theta_{j-1})} \GE \frac{\sound(\Tilde \Theta_{j-1}, f_{j})}{\comp(\Tilde \Theta_{j-1})} - \eps \GE \alpha_{\lim} - \eps\;, \label{eq:cba}
\end{align}
as desired.

For the claim about $\Tilde{F}$, assume for sake of contradiction that there is a $\eps' > 0$ such that 
\[\underset{f \sim \tilde{F}}{\mathbb E} \left[\frac{\sound(\tilde{\Theta}, f)}{\comp(\tilde{\Theta})}\right] \GE \alpha_{\lim} + \eps'\;.\]
Then, we must have that for all $i \ge 1$,
\[
\frac{\sound(\Tilde \Theta_i, f_{i+1})}{\comp(\Tilde \Theta_i)} \GE \alpha_{\lim} + \eps'\;.
\]
However, if we take the limit in (\ref{eq:cba}) as $\eps \to 0$ and $T \ge (C/\eps)^{|\Theta|}$, we obtain that $\alpha \ge \alpha_{\lim} + \eps'$, a contradiction.
\end{proof}

\begin{remark}
The second claim of \emph{Theorem~\ref{thm:alg-correct}} can also be proved for continuous distributions $\tilde{\mathcal F}$ over $\mathcal F$. In that case, we can approximately discretize $\tilde{\mathcal F}$ by picking representative functions which cover the space of functions in the $\ell_{\infty}$ metric with respect to $s_{\Theta}$. We omit further details.
\end{remark}

\begin{remark}
Although this proof only gives correctness when $T$ is exponential in the size of $\Theta$, in practice our simulation only requires $T = |\Theta|^{O(1)}$ rounds to converge with $\eps \approx 10^{-6}$. Perhaps this suggests that the theoretical analysis can also be improved. 
\end{remark}

\subsubsection{Extension to infinite \texorpdfstring{$\Theta$}{Theta}}\label{sub:theta}

We now briefly discuss how to extend Algorithm~\ref{alg:thresh} to allow $\Theta$ to grow. Let $\Theta_{\text{valid},\eps}$ be the space of all valid configurations of \MAXDICUT\ with completeness at least $\eps$. Assume we also have an oracle $\mathcal O_{\Theta}$ which when given a distribution of rounding functions $\tilde{F}$ outputs the configuration $\theta \in \Theta_{\text{valid},\eps}$ on which $\tilde{F}$ performs the worst. We can then dynamically grow our ``working set'' of configurations~$\Theta$ using the following procedure.

\begin{algorithm}
\caption{$\THRESH$ discovery algorithm (growing $\Theta$)}\label{alg:thresh2}
\begin{algorithmic}[1]

\Procedure{FindThreshFull}{$T$, $T'$}

\State Pick $\Theta_0$ arbitrarily.
\For {$i \in \{1, 2, \hdots, T'-1\}$}
    \State $\tilde{F_i} \leftarrow \textsf{FindThresh}(\Theta_{i-1}, T)$.
    \State $\theta_i \leftarrow \mathcal O_{\Theta}(\Tilde F_i)$
    \State $\Theta_i \leftarrow \Theta_{i-1} \cup \{\theta_{i}\}$.
\EndFor

\State \Return $\textsf{FindThresh}(\Theta_{T'}, T)$
  
\EndProcedure
\end{algorithmic}
\end{algorithm}

Let $\alpha_i$ the performance guarantee of $\tilde{F}_i$ over $\Theta_{i-1}$ and let $\hat{\alpha}_i$ be optimal approximation ratio if $T$ were to tend to $\infty$. Note that $\hat{\alpha}_i$ must monotonically decrease (although non-necessarily strictly). Since each $\hat{\alpha}_i$ is nonnegative, they must have a limit $\hat{\alpha}_{\lim}$. Via an argument similar\footnote{This further requires that the family of functions $\mathcal F$ is uniformly continuous: that is small changes in the configurations imply that the rounding functions do not change much. This is true for uniformly bounded, piecewise linear functions.} to Theorem~\ref{thm:alg-correct}, we can take an $\delta$-net over the configuration space $\Theta_{\text{valid}, \eps}$ and argue that if both $\theta_i$ and $\theta_j$ are in the same region of the $\delta$-net, then $\Tilde F_j$ must perform with a ratio at least $\hat{\alpha}_{\lim} -\eps$ on all configurations\footnote{In practice, the distribution of functions also does well on instances with completeness less than $\eps$.} in $\Theta_{\text{valid},\eps}$. In particular, we can guarantee that when $T'$ is sufficiently large, then nearly all $\tilde{F_i}$'s with $i \in \{T'/2, \hdots, T'\}$ are near-optimal distributions. This proves to be an adequate guarantee for practical simulation.

\subsubsection{Implementation details}\label{sub:details}

We now discuss the implementation details for how the ``raw'' $\THRESH$ scheme was generated as well as details of how the ``clean'' $\THRESH$ scheme was derived from it. %

\paragraph{The raw distribution.} Overall, the algorithm for discovering the ``raw'' $\THRESH$ distribution was implemented in Python (version 3.10).

The oracle $\mathcal O_{\mathcal F}$ is computed using the SciPy library's {\sf minimize} routine~\cite{2020SciPy-NMeth} which finds a locally maximum rounding function $f$ when given a starting function $\hat{f} : S \to (-\infty, \infty)$ as input. For numerical stability, we assume that all thresholds are in the range $[-2, 2]$. We compute $\mathcal O_{\Theta}$ by computing $\sound(\theta, \Tilde F)$ for $\theta$'s in a suitably spaced grid and then calling {\sf minimize} on the worst grid point to further tune the parameters.

The $\sound$ routine was computed using Genz's numerical algorithms for approximate multivariate normal integration \cite{genz1992numerical,genz1993comparison} which is bundled with SciPy.  The linear programming routines were implemented using CVXPY~\cite{diamond2016cvxpy,agrawal2018rewriting} as a wrapper around the ECOS solver~\cite{domahidi2013ecos}. 

In practice, we found that the convergence was more stable by additionally adding $\flip(f_i)$ to $F_i$ in Algorithm~\ref{alg:thresh}.  Likewise, in Algorithm~\ref{alg:thresh2}, it was best to add $\flip(\theta_i)$ along with $\theta_i$.

Perhaps the most sensitive part of this algorithm is the choice of the initial $\Theta_0$ in Algorithm~\ref{alg:thresh2}. We found it best to set $\Theta_0$ to be a near-optimal hard distribution. With this choice, it only took $T' \approx 150$. In practice, we did not aim for a fixed $T$ in Algorithm~\ref{alg:thresh}, but rather a more complicated stopping criteria based on how fast $\alpha_i$ is stabilizing. This roughly translates to $T \ll 100$. In total, it took a few hours of single-core computation on a standard desktop computer to find the $\THRESH$ function described in Section~\ref{sec:scheme}. However, as mentioned in Section~\ref{sub:theta}, the worst-case performance of the distribution $\Tilde F_i$ is not monotone in $i$, so it took a few instances of trial and error (i.e., run for a few more iterations) until the worst-case performance was satisfactory.

We further remark that routines similar to the ones described in this section were used to discover (approximately) the configurations used to prove the upper bound on \MAXDICUT\ in Section~\ref{sec:upper} (in this case $\Theta_0$ was seeded to be a fixed $\eps$-spaced grid).

\paragraph{The clean distribution.} 
Inspecting the 39 functions of the raw distribution revealed that they naturally divide into 7 families of functions, with the functions in each family being fairly similar to each other. Taking a weighted average of the functions in each family yielded a scheme with only 7 functions that did almost as well as the original scheme. Further inspection revealed that one of these 7 functions was almost odd, and that the other six functions divide into three pairs in which functions are close to being flips of each other. The first function was made odd by taking the average of the function and its flip. Similarly, the functions in each pair were made flips of each other. This slightly improved the performance ratio obtained. Finally, numerical optimization was used to perform small optimizations. The resulting 7 functions are the ones given in Table~\ref{tbl:clean-dicut}. The final performance ratio obtained was slightly better than the one achieved by the raw distribution. The computations and optimizations were done using MATLAB.

\subsection{Verification using interval arithmetic}

\subsubsection{Sketch of the algorithm}

From now on, we use $\Tilde{F}$ to refer to the clean distribution of 7 functions from the previous subsection. To prove that the claimed distribution $\Tilde{F}$ of rounding functions achieves an approximation ratio of at least $\alpha$ for \MAXDICUT, we need to show that
\[
\forall \theta, \left(\comp(\theta)\neq0 \;\implies\;\frac{\E_{f \sim \Tilde{F}}[\sound(\theta, f)]}{\comp(\theta)} \geq \alpha\right),
\]
or equivalently
\[
\forall \theta, \qquad \E_{f \sim \Tilde{F}}[\sound(\theta, f)] - \alpha \cdot \comp(\theta) \GE 0\;.
\]

Note that in the above expression, $\comp(\theta)$ only involves simple arithmetic operations, and $\E[\sound(\theta, f)]$ is a weighted sum of two-dimensional Gaussian integrals, while $\theta$ takes value in $[-1,1]^3$, modulo the triangle inequalities.

To rigorously verify the inequality for all configurations, we deploy the technique of \emph{interval arithmetic}. In interval arithmetic, instead of doing arithmetics with numbers, we apply arithmetic operations to intervals. Let $op$ be a $k$-ary operation and $I_1, I_2, \ldots, I_k$ be $k$ intervals, then the interval arithmetic on $op(I_1, I_2, \ldots, I_k)$ will produce an interval $I_{op}$ with the following \emph{rigorous} guarantee: $op(x_1, x_2, \ldots, x_k) \in I_{op}$ for every $(x_1, x_2, \ldots, x_k) \in I_1 \times I_2 \times \cdots \times I_k$. By transitivity of set inclusion, if we implement a function $g$ as a composition of such operations in interval arithmetic, then it is guaranteed that the range of $g$ is included in the output interval $I_g$.

This property is useful when it comes to certifying the nonnegativity of $g$. Indeed, if the output interval $I_g$ lies entirely in $[0, \infty)$, then we can establish that $g$ is a nonnegative function on the given input intervals. However, since the computation is usually not exact, to maintain correctness, $I_{g}$ will also contain elements that are not in the range of $g$. In particular, if $g$ attains $0$, then we cannot hope to certify the nonnegativity of $g$ with interval arithmetic unless some very special conditions on $g$ allow for exact evaluation.

Even in the case where $\inf(g) > 0$, $I_g$ may still contain negative elements. For example, if $g = g_1 + g_2$, then $I_g$ might be obtained by adding $I_{g_1}$ and $I_{g_2}$. This will imply that $\sup(I_{g_1}) + \sup(I_{g_2}) \in I_g$, while in reality $g_1$ and $g_2$ may attain maximum/supremum on very different inputs. This issue can be resolved via a simple \emph{divide-and-conquer} algorithm. Whenever the check on $I_g$ is inconclusive, i.e., it contains both positive and negative numbers, then we split one of the input intervals into halves, and recursively apply the same computation to each half. This is like using a microscope: if we cannot see a region clearly, we zoom in to get a better view.%

\begin{algorithm}
\caption{Interval arithmetic verification algorithm}\label{alg:int_1}
\begin{algorithmic}[1]

\Procedure{CheckRatio}{$I_1$, $I_2$, $I_{1,2}$}

\If{$\textsc{CheckValidity}(I_1, I_2, I_{1,2}) = \text{FALSE}$}
\State \Return TRUE
\EndIf
\State $I \gets \textsc{IntervalArithmeticEvaluate}(I_1, I_2, I_{1, 2})$.

\If{$I \subseteq [0, \infty)$}
  \State \Return TRUE
\ElsIf{$I \subseteq (\infty, 0)$}
  \State \Return FALSE
\Else
  \If{$|I_1| = \max(|I_1|, |I_2|, |I_{1, 2}|)$}
      \State Split $I_1$ into two equal-length sub-intervals $I_1 = I_1^l \cup I_1^r$.
      \State \Return $\textsc{CheckRatio}(I_1^l, I_2, I_{1,2}) \wedge \textsc{CheckRatio}(I_1^r, I_2, I_{1,2})$
  \ElsIf{$|I_2| = \max(|I_1|, |I_2|, |I_{1, 2}|)$}
      \State Split $I_2$ into two equal-length sub-intervals $I_2 = I_2^l \cup I_2^r$.
      \State \Return $\textsc{CheckRatio}(I_1, I_2^l, I_{1,2}) \wedge \textsc{CheckRatio}(I_1, I_2^r, I_{1,2})$
  \Else
      \State Split $I_{1, 2}$ into two equal-length sub-intervals $I_{1, 2} = I_{1, 2}^l \cup I_{1, 2}^r$.
      \State \Return $\textsc{CheckRatio}(I_1, I_2, I_{1,2}^l) \wedge \textsc{CheckRatio}(I_1, I_2, I_{1,2}^r)$
  \EndIf
\EndIf
  
\EndProcedure
\end{algorithmic}
\end{algorithm}

The pseudocode of the algorithm is presented in Algorithm~\ref{alg:int_1}. The $\textsc{CheckValidity}$ function checks if there exists a valid configuration in $I_1 \times I_2 \times I_{1,2}$, i.e., a configuration that satisfies all triangle inequalities, and returns true if it does. If $\textsc{CheckValidity}$ returns false, then the algorithm returns true, since in this case the region consists entirely of invalid configurations and there is nothing to check. Otherwise, the algorithm continues to compute an interval $I$, using the $\textsc{IntervalArithmeticEvaluate}$ subroutine, such that
\[
\forall \theta \in I_1 \times I_2 \times I_{1,2}, \quad \E_{f \sim \Tilde{F}}[\sound(\theta, f)] - \alpha \cdot \comp(\theta) \in I.
\]
The algorithm then checks if $I$ is entirely non-negative or entirely negative, in which cases we can decide that either the ratio is achieved over the entire region, or there exists a valid configuration that violates the ratio, and exit the algorithm accordingly. Otherwise, $I$ consists of both positive and negative values, but the negative values may come from evaluation of invalid configurations, or more intrinsically the error produced by interval arithmetic itself. In this case, we subdivide the longest interval into two equal-length sub-intervals and recursively apply the algorithm, as explained earlier.

We implemented this verification algorithm in C using the interval arithmetic library Arb~\cite{johansson2017arb}. Specific advantages of this library is that it has rigorous implementations of the error function~\cite{johansson2019computing} as well as a routine for rigorous numerical integration~\cite{johansson2018numerical}. To speed up the computation, we split the various tasks between cores using GNU Parallel~\cite{Tange2011a}. We obtain the following lemma.

\begin{lemma}\label{lemma:int_arith_cut}
$\Tilde{F}$ achieves an approximation ratio of $\bestDICUTverified$ on all $\MAXDICUT$ configurations with completeness at least $10^{-6}$.
\end{lemma}

We address the requirement on completeness in the next subsection.

\subsubsection{Removing the completeness requirement and a proof of Theorem~\ref{theorem:lower}}\label{section:remove_comp}

As we discussed, interval arithmetic in general cannot certify nonnegativity of a function which attains 0. Unfortunately, the function that we care about, $\E_{f \sim \Tilde{F}}[\sound(\theta, f)] - \alpha \cdot \comp(\theta)$, does attain 0, regardless of the choice of $r$, as the following proposition shows.

\begin{proposition}
Let $\theta = (b_i, b_j, b_{ij})$ be a configuration with $b_i = b_j = b$ and $\rho(\theta) = 1$. Then for any $f$, 
\[
\sound(\theta, f) \EQ \comp(\theta) \EQ 0\;.
\]
\end{proposition}
\begin{proof}
Since $\rho(\theta) = 1$, we have $b_{ij} = b_ib_j + \rho \sqrt{1 - b_i^2}\sqrt{1 - b_j^2} = b^2 + (1 - b^2) = 1$ and
\[
\comp(\theta) \EQ \frac{1 + b_i - b_j - b_{ij}}{4} \EQ \frac{1 + b - b - 1}{4} \EQ 0\;.
\]
For soundness, we have $\sound(\theta, f) = \Phi_{-\rho}(f(b_i), -f(b_j)) = \Phi_{-\rho}(f(b), -f(b))$. Since $\rho = 1$, this is equal to $\Pr_{X \sim N(0, 1)}[X \leq f(b) \wedge -X \leq -f(b)] = \Pr_{X \sim N(0, 1)}[X = f(b)] = 0$.
\end{proof}

Luckily, on configurations with small completeness, it is known that independent rounding, which assigns true to each variable independently with probability 1/2, does very well. Indeed, this rounding scheme satisfies each $\MAXDICUT$ constraint with probability 1/4 on every configuration. This implies that $\Tilde{F}$ combined with the independent rounding will achieve a good approximation ratio over all $\DICUT$ configurations.

\begin{proof}[Proof of Theorem~\ref{theorem:lower}]
Consider the rounding algorithm where we use the $\THRESH$ rounding scheme $\Tilde{F}$ with probability $(1 - 10^{-5})$ and independent rounding with probability $10^{-5}$. We show that this algorithm achieves a ratio of $\bestDICUT$ on all configurations of $\MAXDICUT$.

Let $\theta$ be a $\DICUT$ configuration. If $\comp(\theta) \geq 10^{-6}$, then by Theorem~\ref{lemma:int_arith_cut}, we achieve a ratio of at least $\bestDICUTverified \times (1 - 10^{-5}) > \bestDICUT$. If $\comp(\theta) < 10^{-6}$, then independent rounding contributes a soundness of $0.25 \times 10^{-5} = 2.5 \times 10^{-6} > \bestDICUT\cdot\comp(\theta)$.
\end{proof}

\subsubsection{Further optimizations}
To further speed up the computation, we compute partial derivatives of $\E_{f \sim \Tilde{F}}[\sound(\theta, f)] - \alpha \cdot \comp(\theta)$, and reduce an interval to its boundary point if the corresponding partial derivative is nonnegative or nonpositive.

For example, if we have
\[
\forall \theta \in I_1 \times I_2 \times I_{1,2}, \quad  \frac{\partial}{\partial b_1} \left(\E_{f \sim \Tilde{F}}[\sound(\theta, f)] - \alpha \cdot \comp(\theta)\right) \geq 0\;,
\]
and $I_1 = [l, r]$, then to certify the nonnegativity of $\E_{f \sim \Tilde{F}}[\sound(\theta, f)] - \alpha \cdot \comp(\theta)$, it is sufficient to check
\[
\forall \theta \in \{l\} \times I_2 \times I_{1,2}, \quad  \E_{f \sim \Tilde{F}}[\sound(\theta, f)] - \alpha \cdot \comp(\theta) \GE 0\;.
\]
We remark that we only perform this optimization in regions that are entirely valid, i.e., consisting only of valid configurations. This is because otherwise we may reduce the region to an invalid subregion, on which the program returns true without checking the ratio.
\subsubsection{Implementation details}

To compute the soundness, we need to evaluate bivariate Gaussian distributions of the form $\Phi_{\rho}(t_1, t_2)$. However, Arb only has implementation of one-dimensional integration. To overcome this, we use the following formula from~\cite{DW90}, which transforms $\Phi_{\rho}(t_1, t_2)$ into a one-dimensional integral:
\[
\Phi_{\rho}(t_1, t_2) \EQ \frac{1}{2\pi}\int_0^\rho \frac{1}{\sqrt{1 - r^2}}\exp\left(-\frac{t_1^2 - 2rt_1t_2 + t_2^2}{2(1 - r^2)}\right) dr + \Phi(t_1)\Phi(t_2)\;. 
\]

Another potential issue is numerical stability. Computing $\rho$ from $(b_i, b_j, b_{ij})$ involves division by $\sqrt{(1 - b_i^2)(1 - b_j^2)}$, which can be unstable when $b_i$ or $b_j$ is close to $\pm 1$. In the actual implementation, we overcome this by representing a configuration using $(b_i, b_j, \rho)$.

\section{A new approximation algorithm for \texorpdfstring{\MAXAND}{MAX 2-AND}}\label{sec:lower-AND}

Recall that $\THRESH$ rounding schemes for \MAXAND\ are nearly identical to those for \MAXDICUT, except that the rounding schemes for \MAXAND\ are required to be \emph{odd} functions. It is easy to enforce in the discovery algorithm that the family of piecewise-linear functions we consider are odd (in fact, the oracle runs quicker as the number of free parameters is cut in half). Empirically, we found a ``raw'' distribution of 15 rounding functions which attains a ratio of approximately $0.8741$. %
Using a clean-up procedure similar to that for \MAXDICUT, we were able to simplify it to another distribution $\Tilde{F'}$ with only $3$ functions. See Table~\ref{tbl:clean-2and} for details.

\begin{table}%
\begin{center}
\begin{tabular}{cc}
\scriptsize
\vspace*{-1.5cm}
\begin{tabular}[b]{r|r|r|r}
 &  \hfill\nobreak $f_1$  \hfill\nobreak  &  \hfill\nobreak $f_2$  \hfill\nobreak  &  \hfill\nobreak $f_3$  \hfill\nobreak  \\
\hline
\hfill\nobreak prob \hfill\nobreak  & $ 0.998105$ & $ 0.001126$ & $ 0.000769$\\
\hline
$-1.000000$ & $-1.585394$ & $ 0.934459$ & $ 0.163540$ \\
$-0.700000$ & $-0.870350$ & $ 0.443616$ & $-0.212976$ \\
$-0.450000$ & $-0.512239$ & $ 0.675617$ & $-1.435794$ \\
$-0.300000$ & $-0.332896$ & $-1.446206$ & $ 0.289432$ \\
$-0.250000$ & $-0.274526$ & $-1.495506$ & $ 2.000000$ \\
$-0.179515$ & $-0.193131$ & $-0.382870$ & $-0.492446$ \\
$-0.164720$ & $-0.176869$ & $ 0.015196$ & $-0.933550$ \\
$-0.100000$ & $-0.107901$ & $ 2.000000$ & $-1.568231$ \\
\hline
$ 0.000000$ & $ 0.000000$ & $ 0.000000$ & $ 0.000000$ \\
\hline
$ 0.100000$ & $ 0.107901$ & $-2.000000$ & $ 1.568231$ \\
$ 0.164720$ & $ 0.176869$ & $-0.015196$ & $ 0.933550$ \\
$ 0.179515$ & $ 0.193131$ & $ 0.382870$ & $ 0.492446$ \\
$ 0.250000$ & $ 0.274526$ & $ 1.495506$ & $-2.000000$ \\
$ 0.300000$ & $ 0.332896$ & $ 1.446206$ & $-0.289432$ \\
$ 0.450000$ & $ 0.512239$ & $-0.675617$ & $ 1.435794$ \\
$ 0.700000$ & $ 0.870350$ & $-0.443616$ & $ 0.212976$ \\
$ 1.000000$ & $ 1.585394$ & $-0.934459$ & $-0.163540$ \\
\end{tabular}
\vspace*{1cm} &
\null\vspace*{0.5cm}
\includegraphics[width=3.5in]{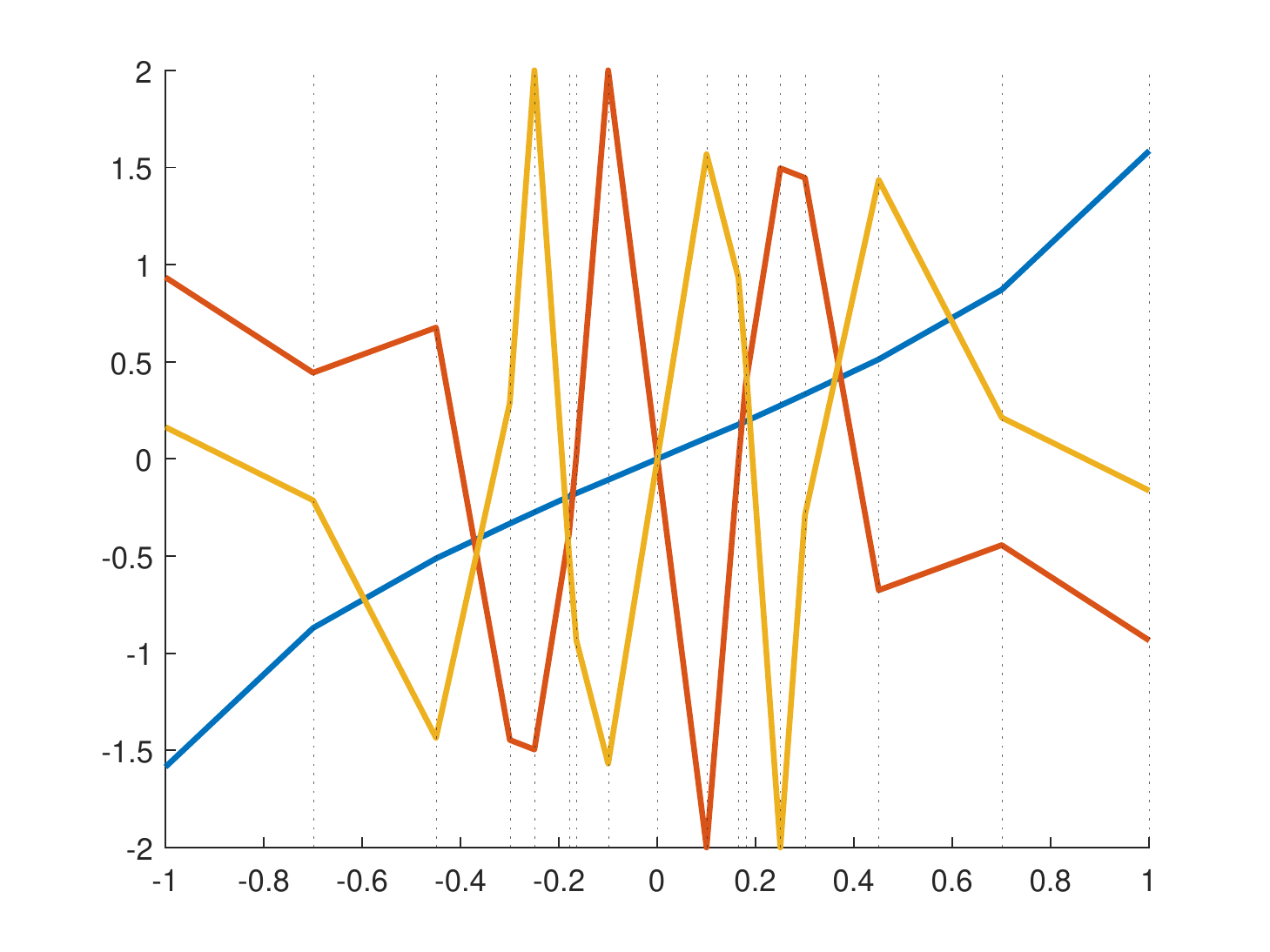}
\end{tabular}
\end{center}
\caption{A $\THRESH$ rounding scheme that gives a rigorously verified approximation ratio of at least $\bestAND$ for \MAXAND. (The actual ratio is probably about $0.874202$.) The scheme uses three piecewise-linear odd rounding functions $f_1,f_2,f_3$ defined on 17 control points. A plot of the functions is given on the right.}\label{tbl:clean-2and}
\end{table}

\begin{figure}
\begin{center}
\includegraphics[width=5in]{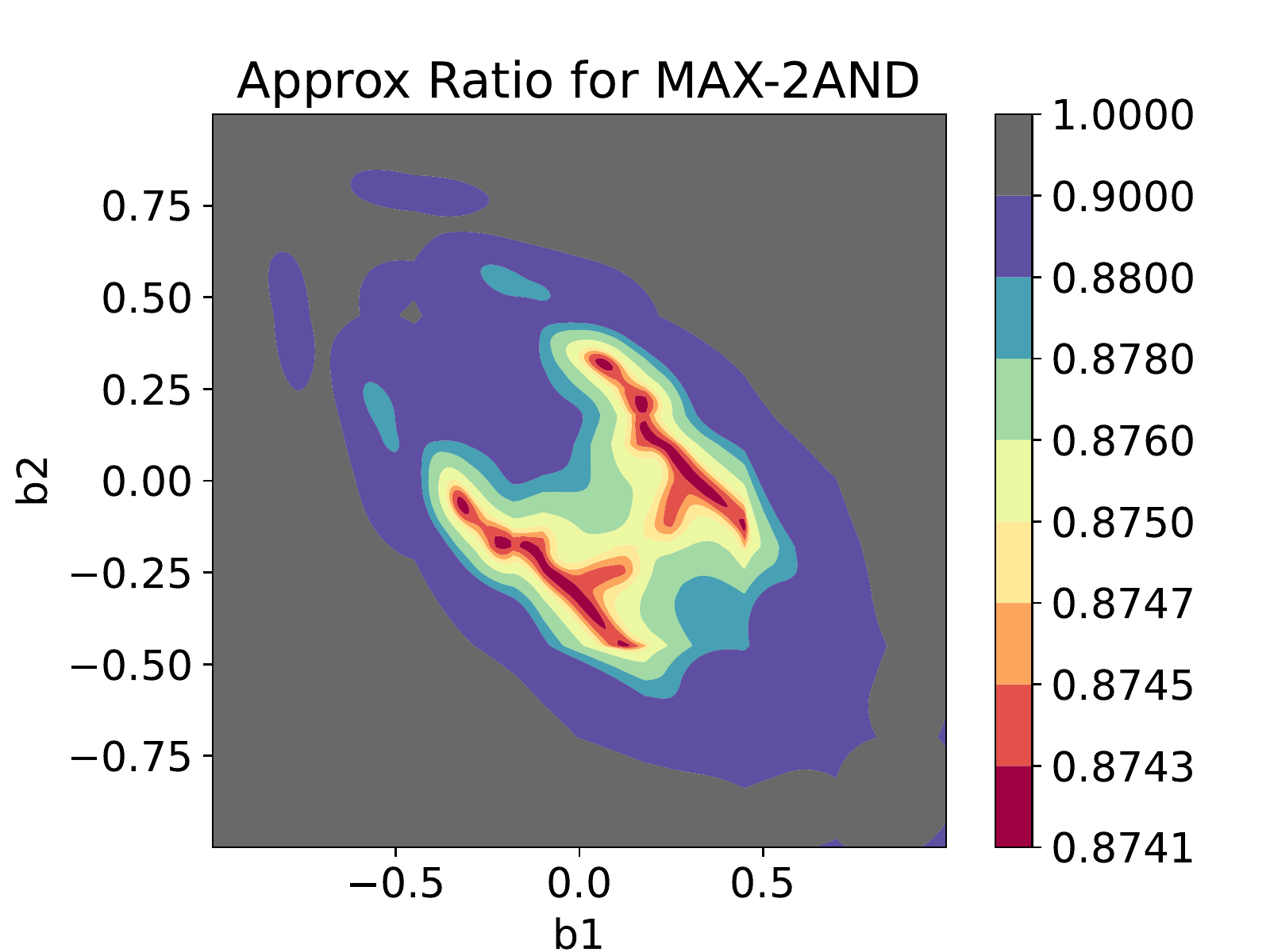}
\end{center}

\caption{This plot is a contour plot of the performance of the $\THRESH$ scheme for \MAXAND\ with $3$ piecewise-linear rounding functions for various choices of $b_1$ and $b_2$ (with an approximately worst-case choice of $b_{12}$) selected.}
\end{figure}

Using the same interval arithmetic algorithm used for \MAXDICUT, we obtain the following result.
\begin{lemma}\label{lemma:int_arith_and}
$\Tilde{F'}$ achieves an approximation ratio of $\bestANDverified$ on all $\MAXAND$ configurations with completeness at least $10^{-6}$. 
\end{lemma}

We can then use the same proof idea as that in Section~\ref{section:remove_comp} to get rid of the completeness requirement and obtain the lower bound of $\bestAND$ for \MAXAND , as claimed in Theorem~\ref{theorem:2and}.

\section{Conclusion}\label{sec:concl}

We used a ``computational lens'' to obtain a much better, and an almost complete, understanding of the \MAXDICUT\ and \MAXAND\ problems. Insights gained from numerical experiments yielded a completely analytical new upper bound for \MAXDICUT\ that can be verified by hand (see Section~\ref{sec:upper}), as well as new lower bounds, i.e., new approximation algorithms, for \MAXDICUT\ and \MAXAND, for which we obtain a rigorous computer-assisted analysis (see Section~\ref{sec:lower} and Section~\ref{sec:lower-AND}).

We have established that the \MAXDICUT\ problem has its own approximation ratio by strictly separating it from \MAXAND\ and \MAXCUT\ (assuming the unique games conjecture). Fundamental to our approach was the use of algorithmic discovery to identify both difficult instances of \MAXDICUT\ and \MAXAND\ as well as discovering $\THRESH$ rounding schemes which improve on the $20+$ year state of the art. %

As discussed in Section~\ref{sec:lower}, assuming the unique games conjecture and Austrin's positivity conjecture, the optimal $\THRESH$ schemes\footnote{Or more precisely a limiting sequence of finite, bounded $\THRESH$ schemes.} achieve $\aDC$ and $\aAND$ for \MAXDICUT\ and \MAXAND, respectively. We demonstrated a computational procedure which helps us to approximate $\aDC$ and $\aAND$ to greater precision than previously known. However, a proper theoretical understanding is still missing. In particular:

\textbf{Theoretical understanding of the optimal $\THRESH$ scheme.} Currently, we lack a satisfactory explanation of why the secondary functions in the currently best-known \MAXDICUT\ and \MAXAND\ $\THRESH$ schemes take on the shapes they do. Perhaps one can prove that the optimal functions must satisfy particular constraints (such as in the calculus of variations), or at least provide a satisfactory understand of the second-order affect these functions have.

\textbf{Theoretical understanding of the hardest configurations.} Likewise, we do not understand the structure of the hardest distributions of configurations for \MAXDICUT\ and \MAXAND. Appendices~\ref{A-Upper} and \ref{A-Upper-AND} show that some rather complex distributions appear to give increasingly better upper bounds for $\aDC$ and $\aAND$. Would it be possible to theoretically describe what the hardest configurations are? It is not clear whether the hardest distribution should even have finite support. Properly describing the hardest distributions would also resolve Austrin's positivity conjecture.

\pagebreak

\appendix

\section{Translating into UG-hardness}\label{A-reduction}

Some of the notations in this section are borrowed from~\cite{Austrin07,Austrin10}. We remark that the techniques in this section are standard and well-known, and only small modifications to that in~\cite{Austrin10}, namely, we drop the requirement that a rounding function has to be odd.

\subsection{Preliminary: (Extended) Majority is Stablest}

We recall some definitions from the analysis of Boolean functions (e.g., \cite{Austrin10,o2014analysis}). Let $B_q^n$ be the probability space over $\{-1, 1\}^n$ where each bit is independently set to $-1$ with probability~$q$ and to $1$ with probability $1 - q.$. Let $U_q(1) = \sqrt{\frac{q}{1 - q}}$ and $U_q(-1) = -\sqrt{\frac{1 - q}{q}}$, and for any $S \subseteq [n]$, let $U_q^S(x_1, \ldots, x_n) = \prod_{i \in S} U_q(x_i)$. It is easy to verify  (c.f., Proposition 2.7 of~\cite{Austrin10}) that 
\[\{U_q^S: B_q^n \to \mathbb{R} \mid S \subseteq [n] \}\]
is an orthonormal basis for real-valued functions on $B_q^n$ with respect to the inner product defined via expectation.  We define the Fourier coefficients of $f$ as $\hat{f}_S = \underset{\bx \sim B_q^n}{\E}[f(\bx) U_q^S(\bx)]$. Note that these Fourier coefficients form a basis decomposition:
\[
f = \sum_{S \subseteq [n]} \hat{f}_S U_q^S.
\]

In our application, we are also interested in computing correlation of two functions with different biases.

\begin{definition}[Definition 2.13, \cite{Austrin10}]
Let $f: B_{q_1}^n \to \mathbb{R}$ and $g: B_{q_2}^n \to \mathbb{R}$. The $\rho$-correlation between $f$ and $g$ is defined as
\[
\mathbb{S}_\rho(f, g) := \E[f(\bx)g(\by)],
\]
where $\bx \sim B_{q_1}^n$, $\by \sim B_{q_2}^n$, and furthermore the $i$-th coordinate of $\bx$ and the $i$-th coordinate of $\by$ has correlation $\rho$, i.e., $\frac{\E[x_iy_i] - \E[x_i]\E[y_i]}{\sqrt{(1 - \E[x_i]^2)(1 - \E[y_i]^2)}} = \rho$.
\end{definition}

\begin{definition}
Let $f: B_q^n \to \mathbb{R}$ and $k \in [n]$. The $k$-low-degree influence of coordinate $i$ on $f$ is defined as
\[
\Inf_i^{\leq k}[f]:= \sum_{S: i \in S\subseteq[n], |S| \leq k} \hat{f}_S^2.
\]
\end{definition}

It is straightforward from the definition that $\Inf_i^{\leq k}$ is convex.

\begin{proposition}\label{prop:Lsize}
Let $f: B_q^n \to [-1, 1]$. For any $\eta > 0$ and $k \in [n]$, we have
\[
\left|\left\{i \in [n] \mid \Inf_i^{\leq k}[f] > \eta \right\}\right| \leq \frac{k}{\eta}.
\]
\end{proposition}
\begin{proof}
We have
\[
    \sum_{i = 1}^n\Inf_i^{\leq k}[f] = \sum_{i = 1}^n\sum_{\substack{S: i \in S\subseteq[n], \\ |S| \leq k}} \hat{f}_S^2 
     = \sum_{|S| \leq k} |S| \hat{f}_S^2 \leq k \cdot \sum_{|S| \leq k} \hat{f}_S^2 \leq k.
\]
The proposition follows immediately.
\end{proof}

It turns out that for functions with small low-degree influences, the extremal behavior of their $\rho$-correlations is characterized by threshold functions in Gaussian space.

\begin{theorem}[Corollary 2.19, \cite{Austrin10}]\label{thm:majority_is_stablest}
For any $\epsilon > 0$, there exist $k \in \mathbb{N}$ and $\eta > 0$ such that for all $f: B_{q_1}^n \to \mathbb{R}$ and $g: B_{q_2}^n \to \mathbb{R}$ satisfying $\min(\Inf_i^{\leq k}[f], \Inf_i^{\leq k}[g]) \leq \eta$ for every $i \in [n]$, we have
\[
4\Phi_{-|\rho|}(t_1, t_2) - \epsilon \leq \mathbb{S}_{\rho}(f, g) - \E[f] - \E[g] + 1 \leq 4\Phi_{|\rho|}(t_1, t_2) + \epsilon,
\]
where $t_1 = \Phi^{-1}\left(\frac{1 - \E[f]}{2}\right)$ and $t_2 = \Phi^{-1}\left(\frac{1 - \E[g]}{2}\right)$.
\end{theorem}

\subsection{UG-Hardness via PCP}

For any permutation $\pi: [L] \to [L]$ and vector $\bx = (x_1, \ldots, x_L)\in \mathbb{R}^L$, let $\pi \bx$ be the vector $(x_{\pi(1)}, \ldots, x_{\pi(L)})$. Given a distribution of configurations $\Tilde{\Theta}$, consider the following PCP protocol $\mathsf{Verifier}_{\Tilde{\Theta}}(I, F)$ (c.f., Algorithm 1 of \cite{Austrin10}):
\begin{itemize}
    \item Input: A Unique Games instance $I = (G, L, \Pi)$, and a set of functions $F = \{f_v: \{-1, 1\}^L \to \{-1, 1\} \mid v \in V(G)\}$.
    \item Choose $v \sim V(G)$ uniformly at random.
    \item Choose two edges incident to $v$ uniformly at random. Call them $e_1 = \{v, u_1\}$, $e_2 = \{v, u_2\}$.
    \item Sample $\theta = (b_1, b_2, b_{12})$ from $\Tilde{\Theta}$.
    \item Independently for every $i \in [L]$, sample $x^{(1)}_i, x^{(2)}_i \sim \{-1, 1\}$ such that $\E[x^{(1)}_i] = b_1, \E[x^{(2)}_i] = b_2, \E[x^{(1)}_ix^{(2)}_i] = \rho$. Let $\bx^{(1)} = (x^{(1)}_1, x^{(1)}_2, \ldots, x^{(1)}_R)$ and $\bx^{(2)} = (x^{(2)}_1, x^{(2)}_2, \ldots, x^{(2)}_R)$.
    \item Compute $\mu_1 = f_{u_1}(\pi_{e_1}^{u_1}\bx^{(1)})$, $\mu_2 = f_{u_2}(\pi_{e_2}^{u_2}\bx^{(2)})$.\
    \item Accept with probability $\DICUT(\mu_1, \mu_2)$.
\end{itemize}

\begin{lemma}[Completeness, c.f., Lemma 5.2 of \cite{Austrin10}]
If $\Val(I) \geq 1 - \eta$, then there exists $F$ such that $\mathsf{Verifier}_{\Tilde{\Theta}}(I, F)$ accepts with probability at least $(1 - 2\eta) \cdot \comp(\Tilde{\Theta})$.
\end{lemma}
\begin{proof}
Since $\Val(I) \geq 1 - \eta$, there exists an assignment $A$ such that $\Val(I, A) \geq 1 - \eta$. For any $v \in V(G)$, let $f_v: \{-1, 1\}^L \to \{-1, 1\}$ be the dictatorship function $(x_1, x_2, \ldots, x_L) \mapsto x_{A(v)}$, and let $F = \{f_v | v \in V(G) \}$. If $\mathsf{Verifier}_{\Tilde{\Theta}}(I, F)$ chooses two edges $e_1, e_2$ that are both satisfied by $A$, then we have
\[
\mu_1 = f_{u_1}(\pi_{e_1}^{u_1}\bx^{(1)}) = (\pi_{e_1}^{u_1}\bx^{(1)})_{A(u_1)} = x^{(1)}_{\pi_{e_1}^{u_1}(A(u_1)))} = x^{(1)}_{A(v)},  
\]
and similarly $\mu_2 = x^{(2)}_{A(v)}$. It follows that 
\begin{align*}
&\Pr[\mathsf{Verifier}_{\Tilde{\Theta}}(I, F) \textrm{ accepts}] \\
 \geq &\, \Pr[e_1, e_2 \textrm{ both satisfied by A}] \cdot \Pr[\mathsf{Verifier}_{\Tilde{\Theta}}(I, F) \textrm{ accepts} \mid e_1, e_2 \textrm{ both satisfied by A}] \\
 \geq &\, (1 - 2 \eta) \cdot \E_{\theta \sim {\Tilde{\Theta}}}\left[\DICUT\left(x^{(1)}_{A(v)}, x^{(2)}_{A(v)}\right)\right] \\
 \geq &\, (1 - 2 \eta) \cdot \E_{\theta \sim {\Tilde{\Theta}}}\left[\DICUT\left(b_1, b_2\right)\right] \\
 = &\, (1 - 2 \eta) \cdot \comp(\Tilde{\Theta}) . \qedhere
\end{align*}
\end{proof}

\begin{lemma}[Soundness, c.f., Lemma 5.3 of \cite{Austrin10}]
For any $\epsilon > 0$ there exists $\gamma > 0$ such that, if $\Val(I) \leq \gamma$, then for any $F$, $\mathsf{Verifier}_{\Tilde{\Theta}}(I, F)$ accepts with probability at most $\max_h \sound({\Tilde{\Theta}}, h) + \epsilon$.
\end{lemma}
\begin{proof}

Fix some $\epsilon > 0$. We need to find some $\gamma > 0$ with the following property: if there exists $F = \{f_v \mid v \in V(G)\}$ such that $\mathsf{Verifier}_{\Tilde{\Theta}}(I, F)$ accepts with probability greater than $\max_h \sound({\Tilde{\Theta}}, h) + \epsilon$, then $\Val(I) > \gamma$. Assume the existence of such $F$, it suffices to show that $\Val(I)$ is lower-bounded by some constant only depending on $\epsilon$.

For $v \in V(G)$ and $b \in (-1, 1)$, we define $g^b_v: B^L_{(1-b)/2} \to [-1, 1]$ as
\[
g^b_v(\bx) = \E_{e = \{v, u\} \in E(G)}[f_{u}(\pi_{e}^{u}\bx)].
\]

Notice that the family of functions $\{g^b_v\}$ naturally lead to the family of thresholds $h_v(b) := \Phi^{-1}\left( \frac{1+\E_\bx[g^b_v(\bx)]}{2}\right)$, under which a variable with bias $b$ has expected value $\E_\bx[g^b_v(\bx)]$ after rounding. We start by computing the accepting probability of the verifier as follows.
\begin{align*}
    & \Pr[\mathsf{Verifier}_{\Tilde{\Theta}}(I, F) \textrm{ accepts}] \\
    = & \, \E_{\theta = (b_1, b_2, b_{12}), v, u_1, u_2, \bx_1, \bx_2}\left[\DICUT(\mu_1, \mu_2)\right] \\
    = & \, \E_{\theta = (b_1, b_2, b_{12}), v, u_1, u_2, \bx_1, \bx_2}\left[\frac{1 + \mu_1 - \mu_2 - \mu_1\mu_2}{4}\right] \\
    = & \, \E_{\theta = (b_1, b_2, b_{12}), v, u_1, u_2, \bx_1, \bx_2}\left[\frac{1 + f_{u_1}(\pi^{u_1}_{e_1}(\bx^{(1)})) - f_{u_2}(\pi^{u_2}_{e_2}(\bx^{(2)})) - f_{u_1}(\pi^{u_1}_{e_1}(\bx^{(1)}))f_{u_2}(\pi^{u_2}_{e_2}(\bx^{(2)}))}{4}\right] \\
    = & \, \E_{\theta = (b_1, b_2, b_{12}), v, \bx_1, \bx_2}\left[\frac{1 + g^{b_1}_v(\bx_1) - g^{b_2}_v(\bx_2) - g^{b_1}_v(\bx_1)g^{b_2}_v(\bx_2)}{4}\right] \\
    = & \, \E_{\theta = (b_1, b_2, b_{12}), v}\left[\frac{1 + \E_{\bx_1}[g^{b_1}_v(\bx_1)] - \E_{\bx_2}[g^{b_2}_v(\bx_2)] - \mathbb{S}_{\rho(\theta)} (g^{b_1}_v, g^{b_2}_v)}{4}\right],
\end{align*}
On the other hand, we have
\begin{align*}
    & \max_h \sound(\Tilde{\Theta}, h) + \epsilon\\
    \geq & \, \E_v[\sound(\Tilde{\Theta}, h_v)] + \epsilon\\
    = & \, \E_{\theta = (b_1, b_2, b_{12}), v}\left[\sound(\theta, h_v)\right] + \epsilon \\
    = & \, \E_{\theta = (b_1, b_2, b_{12}), v}\left[\Phi_{-\rho(\theta)}(h_v(b_1), -h_v(b_2))\right] + \epsilon \\
\end{align*}

Since we have assumed
\[
\Pr[\mathsf{Verifier}_{\Tilde{\Theta}}(I, F) \textrm{ accepts}] \geq \max_h \sound(\Tilde{\Theta}, h) + \epsilon,
\]
from the above computation it follows that
\[
\E_{\theta, v}\left[\frac{1 + \E_{\bx_1}[g^{b_1}_v(\bx_1)] - \E_{\bx_2}[g^{b_2}_v(\bx_2)] - \mathbb{S}_{\rho(\theta)} (g^{b_1}_v, g^{b_2}_v)}{4}\right] \geq \E_{\theta, v}\left[\Phi_{-\rho(\theta)}(h_v(b_1), -h_v(b_2))\right] + \epsilon.
\]
Simplifying using the fact that $\Phi_{-\rho}(t_1, -t_2) = \Phi(t_1) - \Phi_{\rho}(t_1, t_2) = \Phi(-t_2) - \Phi_{\rho}(-t_1, -t_2)$, we obtain
\[
\E_{\theta = (b_1, b_2, b_{12}), v}\left[4\Phi_{\rho(\theta)}(-h_v(b_1), -h_v(b_2)) - \left(\mathbb{S}_{\rho(\theta)} (g_v^{b_1}, g_v^{b_2})) - \E_{\bx_1}[g_v^{b_1}(\bx_1)] - \E_{\bx_2}[g_v^{b_2}(\bx_2)] + 1 \right)\right] \geq 4\epsilon.
\]
We can therefore find some $\theta = (b_1, b_2, b_{12})$ such that
\[
\E_{v}\left[4\Phi_{\rho(\theta)}(-h_v(b_1), -h_v(b_2)) - \left(\mathbb{S}_{\rho(\theta)} (g_v^{b_1}, g_v^{b_2})) - \E_{\bx_1}[g_v^{b_1}(\bx_1)] - \E_{\bx_2}[g_v^{b_2}(\bx_2)] + 1 \right)\right] \geq 4\epsilon.
\]
Each term in the above expectation is bounded by some absolute constant, so we can find $C > 0$ such that
\[
\left|4\Phi_{\rho(\theta)}(-h_v(b_1), -h_v(b_2)) - \left(\mathbb{S}_{\rho(\theta)} (g_v^{b_1}, g_v^{b_2})) - \E_{\bx_1}[g_v^{b_1}(\bx_1)] - \E_{\bx_2}[g_v^{b_2}(\bx_2)] + 1 \right)\right| \leq C.
\]

It follows that for at least an $\epsilon$ fraction of $v \in V(G)$, we have
\[
4\Phi_{\rho(\theta)}(-h_v(b_1), -h_v(b_2)) - \left(\mathbb{S}_{\rho(\theta)} (g_v^{b_1}, g_v^{b_2})) - \E_{\bx_1}[g_v^{b_1}(\bx_1)] - \E_{\bx_2}[g_v^{b_2}(\bx_2)] + 1 \right)\geq \frac{\epsilon}{C}.
\]
Let $V_0$ be the set of $v \in V(G)$ that satisfy the above inequality. since the configurations are all positive, we have $\rho
(\theta) \leq 0$, and therefore by Theorem~\ref{thm:majority_is_stablest}, there exist $\eta > 0$ and $k \in \mathbb{N}$ such that, for every $v \in V_0$ there is some $i \in [n]$ with
\[
\Inf_i^{\leq k}[g_v^{b_1}] \geq \min(\Inf_i^{\leq k}[g_v^{b_1}] , \Inf_i^{\leq k}[g_v^{b_2}] ) \geq \eta.
\]
Since $\Inf_i^{\leq k}$ is convex, we also have
\[
\eta \leq \Inf_i^{\leq k}[g_v^{b_1}] = \Inf_i^{\leq k}\left[\E_{e = \{v, u\} \in E(G)}[f_{u}\circ\pi_{e}^{u}]\right]\leq \E_{e = \{v, u\} \in E(G)}[\Inf_i^{\leq k}[f_{u}\circ\pi_{e}^{u}]].
\]
Since $\Inf_i^{\leq k}$ takes value in $[0, 1]$, there is an $\eta / 2$ fraction of $u \sim v$ such that
$\Inf_i^{\leq k}[f_{u}\circ\pi_{e}^{u}] = \Inf_{(\pi_e^{u})^{-1}(i)}^{\leq k}[f_{u}] \geq \eta/2.$ Now let
$L_1(v) = \{i \in [n] \mid \Inf_i^{\leq k}[g_v^{b_1}] \geq \eta \}$ and $L_2(v) = \{i \in [n] \mid \Inf_i^{\leq k}[f_v^{b_1}] \geq \eta / 2\}$. By Proposition~\ref{prop:Lsize}, we have $|L_1(v)| \leq \frac{k}{\eta}$ and $|L_2(v)| \leq \frac{2k}{\eta}$, and by union bound $|L_1(v) \cup L_2(v)| \leq \frac{3k}{\eta}$.

Now consider the following labeling strategy for $I$: for every $v \in V(G)$, if $L_1(v) \cup L_2(v)$ is non-empty, then choose a label $A(v) \in L_1(v) \cup L_2(v)$ uniformly at random, otherwise choose $A(v) \in [R]$ uniformly at random. By our analysis above, if we choose an edge $e = (u, v)$ with $v \in V_0$, then there is at least $\epsilon \cdot \eta / 2$ probability such that there is some $i \in L_1(v)$ with $\pi_e^v(i) \in L_2(u)$, which our strategy will then find with probability at least $1 / (3k / \eta)^2$, so $\Val(I, A)$ is at least $\epsilon \cdot \eta / 2 \cdot 1 / (3k / \eta)^2$, which is a constant only depending on $\epsilon$, and the lemma is proven.
\end{proof}

\pagebreak
\section{Possibly improved upper bounds for \texorpdfstring{\MAXDICUT}{MAX DI-CUT}}\label{A-Upper}

In Section~\ref{sec:upper} we presented a simple distribution on three configurations that shows that $\aDC\le 0.8746024732$ assuming UGC. This distribution, which also given in Table~\ref{T-DICUT-1}, used only one pair of biases, $b$ and~$-b$, where $b= 0.1757079639$. The simplicity of this distribution enabled us to rigorously prove that $\aDC\le 0.8746024732$.

Slightly improved upper bounds on $\aDC$ can probably be obtained using more complicated distributions that use two, three or four pairs of biases, as shown in Tables~\ref{T-DICUT-2}, \ref{T-DICUT-3} and~\ref{T-DICUT-4}. However, analyzing the performance of any rounding procedure from $\THRESH^-$ on these distributions is a much harder task that can probably not be done by hand. The bounds given in Tables~\ref{T-DICUT-2} to~\ref{T-DICUT-4} were only verified using non-rigorous numerical optimizations. 

In the simple case of Table~\ref{T-DICUT-1}, the function $s(t_1,t_2)$, where $t_1$ and $t_2$ are the thresholds corresponding to the thresholds $-b$ and $b$, had a unique global maximum. Unfortunately, the corresponding function $s_2(t_{-2},t_{-1},t_1,t_2)$ for the distribution of Table~\ref{T-DICUT-2}, and the corresponding functions for the distributions of Tables~\ref{T-DICUT-3} and~\ref{T-DICUT-4}, also have local maxima that make a rigorous analysis much more difficult. In some of the cases the global maximum is also not unique. (The probabilities are carefully chosen to make several local maxima have the same value.)

More pairs of biases can of course be used but it seems that the improvement obtained would be negligible, as going from one pair of biases to four pairs of biases improved the upper bound by only $2\times 10^{-5}$. We thus conjecture that the (non-rigorous) upper bound $\aDC\le 0.8745794663$ is close to being tight.

\begin{table}[H]
\centering
\begin{tabular}{c@{\quad\quad}c}
$b\EQ 0.1757079639$ &
\begin{tabular}{c@{\quad\quad(\;}r@{\;,\;}r@{\;,\;}c@{\;\;)}}
\multicolumn{1}{c}{probability\quad\quad\quad} & \multicolumn{3}{c}{configuration}\\
\hline
0.3770580402  &  $b$ &  $b$ & $-1+2b$ \\
0.3770580402  & $-b$ &  $-b$ & $-1+2b$ \\
0.2458839196  &  $b$ & $-b$ & $-$0.6876930468 
\end{tabular} 
\end{tabular}
\caption{The distribution over configurations used to obtain the bound $\aDC\le 0.8746024732$ in Section~\ref{sec:upper}. Only one pair of biases is used.}
\label{T-DICUT-1}
\end{table}

\begin{table}[H]
\centering
\begin{tabular}{c@{\quad\quad}c}
\begin{tabular}{c}
$b_1 \;=\;  0.1644279457$ \\
$b_2 \;=\;  0.1797733117$ \\
\end{tabular} &
\begin{tabular}{c@{\quad\quad(\;}r@{\;,\;}r@{\;,\;}c@{\;\;)}}
\multicolumn{1}{c}{probability\quad\quad\quad} & \multicolumn{3}{c}{configuration}\\
\hline
 0.1907744673 & $b_2$ & $b_1$ & $-1+b_1+b_2$ \\
 0.1907744673 & $-b_1$ & $-b_2$ & $-1+b_1+b_2$ \\
 0.1858539509 & $b_2$ & $b_2$ & $-1+2b_2$ \\
 0.1858539509 & $-b_2$ & $-b_2$ & $-1+2b_2$ \\
 0.2371153723 & $b_1$ & $-b_1$ & $-0.6874089540$ \\
 0.0048138957 & $b_1$ & $-b_2$ & $-0.6876719134$ \\
 0.0048138957 & $b_2$ & $-b_1$ & $-0.6876719134$ \\
\end{tabular}
\end{tabular}
\caption{A distribution that uses two pairs of biases that seems to yield an upper bound $\aDC\le 0.8745896786$. (Not verified rigorously.)}
\label{T-DICUT-2}
\end{table}

\begin{table}[H]
\centering
\begin{tabular}{c@{\quad\quad}c}
\begin{tabular}{c}
$b_1 \;=\;  0.1389906477$ \\
$b_2 \;=\;  0.1758192542$ \\
$b_3 \;=\;  0.2016555060$ \\
\end{tabular} &
\begin{tabular}{c@{\quad(\;}r@{\;,\;}r@{\;,\;}c@{\;\;)}}
\multicolumn{1}{c}{probability\;\;\;} & \multicolumn{3}{c}{configuration}\\
\hline
 0.2267479169 & $b_2$ & $b_2$ & $-1+2b_2$ \\
 0.2267479169 & $-b_2$ & $-b_2$ & $-1+2b_2$ \\
 0.0493365471 & $b_2$ & $b_3$ & $-1+b_2+b_3$ \\
 0.0493365471 & $-b_3$ & $-b_2$ & $-1+b_2+b_3$ \\
 0.1001888661 & $b_3$ & $b_1$ & $-1+b_1+b_3$ \\
 0.1001888661 & $-b_1$ & $-b_3$ & $-1+b_1+b_3$ \\
 0.1237266700 & $b_1$ & $-b_2$ & $-0.6873638769$ \\
 0.1237266700 & $b_2$ & $-b_1$ & $-0.6873638769$ \\
\end{tabular}
\end{tabular}
\caption{A distribution that uses three pairs of biases that seems to yield an upper bound $\aDC\le 0.8745810643$. (Not verified rigorously.)}
\label{T-DICUT-3}
\end{table}

\begin{table}[H]
\centering
\begin{tabular}{c@{\quad\quad}c}
\begin{tabular}{c}
$b_1 \;=\;  0.1367092212$ \\
$b_2 \;=\;  0.1726598484$ \\
$b_3 \;=\;  0.1778293053$ \\
$b_4 \;=\;  0.2039443849$ \\
\end{tabular} &
\begin{tabular}{c@{\quad\quad(\;}r@{\;,\;}r@{\;,\;}c@{\;\;)}}
\multicolumn{1}{c}{probability\quad\quad\quad} & \multicolumn{3}{c}{configuration}\\
\hline
 0.0346789517 & $b_2$ & $b_2$ & $-1+2b_2$ \\
 0.0346789517 & $-b_2$ & $-b_2$ & $-1+2b_2$ \\
 0.0371520073 & $b_2$ & $b_3$ & $-1+b_2+b_3$ \\
 0.0371520073 & $-b_3$ & $-b_2$ & $-1+b_2+b_3$ \\
 0.0495233867 & $b_2$ & $b_4$ & $-1+b_2+b_4$ \\
 0.0495233867 & $-b_4$ & $-b_2$ & $-1+b_2+b_4$ \\
 0.0592278650 & $b_3$ & $b_2$ & $-1+b_2+b_3$ \\
 0.0592278650 & $-b_2$ & $-b_3$ & $-1+b_2+b_3$ \\
 0.0953106050 & $b_3$ & $b_3$ & $-1+2b_3$ \\
 0.0953106050 & $-b_3$ & $-b_3$ & $-1+2b_3$ \\
 0.1003411331 & $b_4$ & $b_1$ & $-1+b_1+b_4$ \\
 0.1003411331 & $-b_1$ & $-b_4$ & $-1+b_1+b_4$ \\
 0.0471058388 & $b_1$ & $-b_2$ & $-0.6876148335$ \\
 0.0471058388 & $b_2$ & $-b_1$ & $-0.6876148335$ \\
 0.0766602123 & $b_1$ & $-b_3$ & $-0.6876243954$ \\
 0.0766602123 & $b_3$ & $-b_1$ & $-0.6876243954$ \\
\end{tabular}
\end{tabular}
\caption{A distribution that uses four pairs of biases that seems to yield an upper bound $\aDC\le 0.8745794663$. (Not verified rigorously.)}
\label{T-DICUT-4}
\end{table}

\pagebreak
\section{Possibly improved upper bounds for \texorpdfstring{\MAXAND}{MAX AND}}\label{A-Upper-AND}

Austrin~\cite{Austrin10} gave two upper bound on the best approximation ratio achievable for \MAXAND, assuming UGC. The first used only one non-zero bias and gave an upper bound $\aAND \le 0.87451$. The second used two non-zero biases and gave an upper bound $\aAND \le 0.87435$. We believe that using four non-zero biases distribution given in Table~\ref{T-MAXAND-upper} it is possible to prove that $\aAND\le 0.874247$, but we have not verified it rigorously.

\begin{table}[H]
\centering
\begin{tabular}{c@{\quad\quad}c}
\begin{tabular}{c}
$b_1 \;=\; 0.0726617 $ \\
$b_2 \;=\; 0.165630 $ \\
$b_3 \;=\; 0.248978 $ \\
$b_4 \;=\; 0.317508 $ \\
\end{tabular} &
\begin{tabular}{c@{\quad\quad(\;}c@{\;,\;}c@{\;,\;}c@{\;\;)}}
\multicolumn{1}{c}{probability\quad\quad\quad} & \multicolumn{3}{c}{configuration}\\
\hline
0.00778369 & 0 & $b_4$ & $-1+b_4$ \\
0.264364 & $b_1$ & $b_4$ & $-1+b_1+b_4$ \\
0.050959 & $b_2$ & $b_2$ & $-1+2b_2$ \\
0.0572364 & $b_2$ & $b_3$ & $-1+b_2+b_3$ \\
0.113076 & $b_3$ & $b_1$ & $-1+b_1+b_3$ \\
0.506466 & $b_4$ & 0 & $-1+b_4$
\end{tabular}
\end{tabular}
\caption{A distribution that uses four non-zero biases that seems to yield an upper bound $\aAND\le 0.874247$. (Not verified rigorously.)}
\label{T-MAXAND-upper}
\end{table}

\pagebreak

\bibliographystyle{alpha}
\bibliography{refs}

\end{document}